\documentclass[bj]{imsart}

\RequirePackage{amsthm,amsmath,amsfonts,amssymb}
\RequirePackage[numbers]{natbib}
\RequirePackage[colorlinks,citecolor=blue,urlcolor=blue]{hyperref}
\RequirePackage{graphicx}

\def \hat{\widehat}
\def \tilde{\widetilde}

\def\dcal{\mathcal{D}}

\def \hatG{\widehat{\gamma}^{(L)}}

\def \hatB{\widehat{\beta}^{(L)}}
\def \w{\omega}

\startlocaldefs
\theoremstyle{plain}

\newtheorem{theorem}{Theorem}[section]
\newtheorem{lemma}[theorem]{Lemma}
\newtheorem{condition}[theorem]{Condition}

\newtheorem{proposition}[theorem]{Proposition}
\theoremstyle{remark}

\newtheorem{remark}[theorem]{Remark}

\endlocaldefs

\begin{document}

\begin{frontmatter}
\title{Treatment Effect Estimation with Efficient Data Aggregation}
\runtitle{Treatment Effect Estimation with Efficient Data Aggregation}

\begin{aug}
\author[A]{\inits{P.}\fnms{Snigdha } \snm{Panigrahi}\ead[label=e1]{psnigdha@umich.edu}},
\author[B]{\inits{W.}\fnms{Jingshen} \snm{Wang}\ead[label=e2]{jingshenwang@berkeley.edu}}
\and
\author[A]{\inits{H.}\fnms{Xuming} \snm{He}\ead[label=e3]{xmhe@umich.edu}}
\address[A]{Department of Statistics, University of Michigan, Ann Arbor, USA.
\printead{e1,e3}}

\address[B]{Division of Biostatistics, UC Berkeley, USA.
\printead{e2}}

\end{aug}

\begin{abstract}
Data aggregation, also known as meta analysis, is widely used to
combine knowledge on parameters shared in common (e.g., average treatment effect) between multiple studies.
In this paper, we introduce an attractive data aggregation scheme that pools summary statistics from various existing studies. 
Our scheme informs the design of new validation studies and yields us unbiased estimators for the shared parameters. 
In our setup, each existing study applies a LASSO regression to select a parsimonious model from a large set of covariates.
It is well known that post-hoc estimators, in the selected model, tend to be biased. 
We show that a novel technique called \textit{data carving} yields us a new unbiased estimator by aggregating simple summary statistics from all existing studies.
Our estimator has two key features: (a) we make the fullest possible use of data, from all studies, without the risk of bias from model selection; (b) we enjoy the added benefit of individual data privacy, because raw data from these studies need not be shared or stored for efficient estimation. 
\end{abstract}


\end{frontmatter}

\section{Introduction}

Aggregating knowledge across studies is a popular practice to pool multiple findings, increase precision of shared results, and plan future studies. 
In this paper, we focus on the problem of treatment effect evaluation after adjusting for potential confounders from various independent, existing studies. 
Existing studies in our setup can be pilot studies, preliminary analyses, or previously reported results, all evaluating the same treatment effect. 
Our paper introduces a data aggregation scheme which informs an efficient design for follow-up validation studies, and yields us estimators for the treatment effects by pooling summary statistics from the existing studies. 

In the context of high-dimensional data, the number of potential measured confounders might be very large.
Estimation of treatment effects in the presence of many possible confounders usually calls for a model selection procedure in any individual study.  
Two practical considerations gain prominence in this setup. 
First, full data from the existing studies, e.g., raw data for the confounders, might be unavailable, due to either practical limitations in data-sharing, or due to legitimate  privacy and confidentiality concerns \citep{wolfson2010datashield, cai2021individual}. Second, if model selection is performed on each existing study, post-hoc estimation of treatment effects in the selected model tends to be biased. 
We thus ask:
(i) whether and how data from two or more studies can be aggregated without the risk of selection bias, and (ii) whether the former goal can be achieved by using only summary statistics that are sufficient for data aggregation.

To address the questions raised above, we focus on the very widely used LASSO regression \citep{tibshirani1996regression}, which is simply called the LASSO.
By placing a penalty on the absolute value of regression coefficients, the LASSO shrinks many of them to be exactly zero and leads to a more parsimonious model. 
In each existing study, when the LASSO is used to select a set of covariates as confounders, the usual summary statistics---the first two moments of data required to conduct regression in the selected model---are no longer sufficient for obtaining unbiased estimators through data aggregation. 
It, however, turns out that the summary statistics required by our aggregation scheme assume a fairly simple and compressed form. 
An immediate by-product of this finding is a new guideline for what needs to be reported in  publications for an efficient and unbiased estimation of shared parameters in models selected by the LASSO.

While data aggregation from existing studies can be performed on their own, follow-up validation studies are either necessitated by a regulatory requirement, or pursued to report findings with higher statistical power. 
Researchers in science and public health  usually find it more cost-efficient to design and carry out new validation studies by relying on a synthesis of prior research.
This practice allows researchers to focus on a smaller set of confounding factors, and reduce variability of estimators from any single study at the same time. 
For example, research into cost-efficient trial designs has led to the development of methodology for seamless phase II/III designs in clinical trials.
Some of these approaches have been reviewed by \cite{stallard2011seamless}. 
A more recent Lancet article by \cite{park2021covid} highlights not just the need for well designed clinical trials, but also emphasizes the need that 
these trials are structured as per a master protocol in a coordinated and collaborative manner. 
Back to our problem setup, the validation study $\mathcal{V}$ collects data on a subset of potential confounders that were selected by any of the already existing studies.

We showcase how a recent technique, namely \textit{data carving}, can be used to efficiently aggregate summary statistics from each existing study  for an unbiased estimation of treatment effects. 
Data carving allows us to re-use information from data that has not been been fully exploited by the LASSO at the time of model selection.
A schematic depiction of our aggregation scheme is provided in Figure \ref{fig:schematic}. 
Suppose that we have $K$ existing studies.
We apply data carving to aggregate summary statistics from existing study $k$ with data from the validation study $\mathcal{V}$.
This gives us our new estimator for treatment effects $\widehat{\alpha}_k^{\text{\; carve}}$, which we call the carved estimator.
A simple averaging of the $K$ carved estimators allows us to pool information from the existing studies, and yields us our final aggregated estimator $\tilde{\alpha}$.

\begin{figure}
	\centering
	\includegraphics[width = 0.95\linewidth]{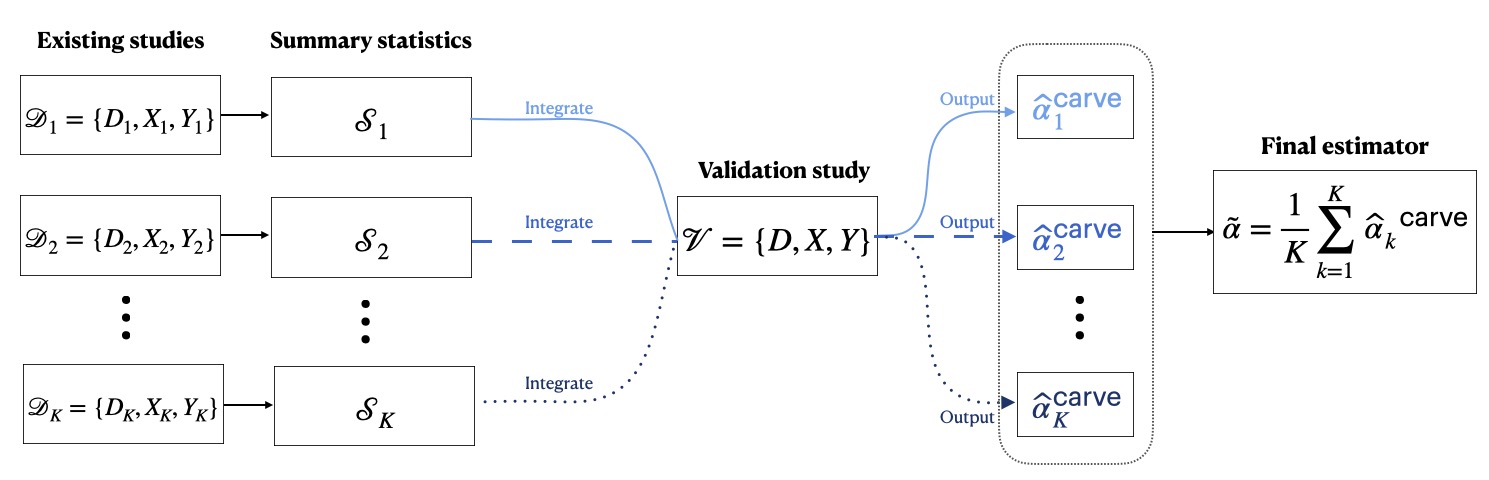}
	\caption{Schematic depiction of efficient and privacy-preserving data aggregation. The exact form for the summary statistics in our scheme and our estimator $\tilde{\alpha}$ is provided in Section 2.}
	\label{fig:schematic}
\end{figure}

\subsection{Related Work} 
A simple estimator in our problem setup is the split-and-aggregate estimator.
The split-and-aggregate estimator is obtained by first (i) refitting the selected model from each existing study on the validation dataset, followed by (ii) averaging the $K$ split-based estimators.
Because the data used for model selection is not re-used at the time of estimation, the resulting estimator does not suffer from the issue of the selection bias.
However, this estimator suffers from a loss of efficiency as it discards all data from the existing studies.
This loss in efficiency is especially noticeable if the validation study is small relative to the existing studies.

Inspired by recent efforts to re-use information from model selection \citep{fithian2014optimal, dwork2015generalization, tian2020prediction}, the carved estimator makes the fullest possible use of data for estimating treatment effects. 
Data carving, as introduced in \citep{fithian2014optimal}, resembles data-splitting in the selection step, because model selection is carried out on only a subset of the full data.
But, carved estimators, unlike those based on data splitting, use the full data instead of the held-out portion alone. 
More precisely, carved estimators are obtained by conditioning on the observed selection event to remove bias from the model selection step.
We point out to previous work by \citep{panigrahi2016integrative, panigrahi2019carving, schultheiss2021multicarving} which deploy data carving to obtain interval estimators for selected regression coefficients. 
More generally, papers by \cite{exact_lasso, tian2018selective, panigrahi2018scalable, gao2020selective, chen2020valid, panigrahi2019approximate, panigrahi2020approximate} have used conditioning as an effective technique to offer post-selection inference through valid p-values, confidence intervals, confidence regions, and credible intervals.

More sophisticated proposals by \cite{tang2020distributed}, \cite{lee2017communication}, and \cite{maity2019communication} aggregate debiased LASSO estimators from distributed datasets. In the same spirit as \cite{lin2010relative}, \cite{cai2021individual} provide a more efficient, inversely variance-weighted debiased LASSO estimator to accommodate cross-study heterogeneity in similar settings. 
The proposed debiased procedures in these papers offer a valid way out to mitigate model selection bias in high-dimensional settings.
But, the debiasing correction involves estimating and inverting a high-dimensional information matrix. 
Moreover, the reduction in bias is not for free, and comes at a price in statistical efficiency. 
Especially, this price is large when the treatment assignment is correlated with noise variables in the model.
For example, \cite{wang2020debiased} have previously illustrated this phenomenon on large observational studies.
Not only does our new estimator avoid estimation (and storage) of high-dimensional matrices for removing bias, but also delivers higher statistical efficiency than existing estimators in the situation described above.

\subsection{Organization of paper}

The remaining paper is organized as follows. 
Section \ref{Sec:framework} formally states our problem setup and introduces the main ideas behind our new estimator through a first example.
In Section \ref{Sec:theory}, we outline the steps to  obtain our carved estimator through an aggregation of summary statistics, and develop the asymptotic theory for the new estimator.
Numerical experiments that support our asymptotic results are presented in Section \ref{Sec:simulation}.
In Section \ref{sec:real}, we highlight the efficacy of our data aggregation and estimation scheme on synthetic data generated from a clinical trial.
Section \ref{Sec:conclusion} concludes our work with some summarizing remarks.

 \section{Methodology}
\label{Sec:framework}
\subsection{Problem setup}

\noindent{\textbf{Basic setup}}.
We consider $K$ independent, existing studies in our setup.
Each study measures $n_k$ independently and identically distributed triples: (i) $Y_{k}\in \mathbb{R}^{n_k \times 1}$ is the observed outcome, (ii) $X_{k} \in \mathbb{R}^{n_k\times p_k}$ represents the  high-dimensional covariates where $p_k$ is allowed to be larger than $n_k$, and (iii) $D_{k}\in\mathbb{R}^{n_k\times s}$ is the $\mathbb{R}^s$-valued treatment variable. 
We denote the data from these studies by
$$\dcal_k = \left\{D_{k}, X_{k}, Y_{k}\right\},$$
for $k\in \{1,2,\cdots, K\}$.
For any $k$, we use $E_k$ to represent an index set of covariates with cardinality less than $p_k$. 
For any vector or matrix $A$ and an index set $E$, let $A_E$ be the sub-vector or sub-matrix containing the columns of $A$ indexed by the set $E$. 
We use the symbol $1_k$ and $0_k$ to represent a $\mathbb{R}^k$-valued vector of all ones and all zeroes, respectively, and use $\|v\|$ to represent the $\ell_2$ norm of a vector unless specified otherwise.

We begin with a simple model to describe the data in our studies.
Suppose that $\varepsilon_k \in \mathbb{R}^{n_k \times 1}$ are random errors that satisfy
$$\mathbb{E}(\varepsilon_k | D_k, X_k) = 0_{n_k},\quad  \text{Cov}(\varepsilon_k | D_k, X_k) =\sigma^2 I,  \quad k =1,2, \ldots, K,$$
where $I$ is the identity matrix.
We assume that the observations in our studies are drawn from the population model 
\begin{align}
\label{popn:modelk}
Y_k =  D_k \alpha + X_{k,E_0}\beta_{E_0} + \varepsilon_k, 
\end{align} 
where $\alpha \in R^s $ measures the treatment effect, $E_0$ indicates the covariates with nonzero effects, and  $\beta_{E_0}$ measures the impact of these covariates. Here, the term treatment effect refers to the coefficients of shared variables across multiple studies. This terminology is consistent with the high dimensional inference literature \cite{belloni2014inference, cattaneo2018inference, belloni2019valid, wang2020debiased} and with our current model setup, where the regression coefficients are often called treatment effects by researchers. 
Later in this section, we show that our introduced scheme generalizes to a variety of linear models, such as models with cross-study heterogeneity.

We can think of the first column of $X_k$ as a vector of all 1's to represent the intercept in the population model \eqref{popn:modelk}.
In practice, we center all the data vectors, and simply assume that $Y_k$ sums to zero, and so does every column of $D_k$ and $X_k$ without a loss of generality. 
Each existing study collects data for an extensive collection of covariates, and the treatment assignment is randomized conditional on these covariates. 
We further assume that the set of observed covariates, $X_k$, in each study include the ones indexed by $E_0$, that is, the covariates with nonzero effects in the population model are measured in each study. 
This is a rather strong assumption, requiring that the set of unobserved confounders is balanced in all the studies. 
Unobserved confounders in more general settings certainly deserve further attention.

Because our existing studies begin with a large number of covariates, we focus on the case where each study has used the LASSO to select a subset of potential confounders.
Denote by $E_k$ the active set of covariates selected by the LASSO in existing study $k$, and let $q_k=|E_k|$.
We assume that $E_0 \subset E_k$, which holds with probability tending to one as $n_k\to \infty$ under appropriate conditions.
Define $$E=\cup_{k=1}^{K} E_k$$ with cardinality equal to $q$.

Besides the existing studies, we also consider a validation study $\mathcal{V} = \left\{ D, X, Y\right\}$, which is independent of $\dcal_k$ for all $k=1,\ldots, K$.
The data in the validation study $\mathcal{V}$ contains $n$ measurements for the same outcome $Y \in \mathbb{R}^{n\times 1}$ and treatment variables $D \in \mathbb{R}^{n \times s}$ as our existing studies. 
Furthermore, $\mathcal{V}$ contains matched measurements for the potential confounders in the subset $E$, that are now represented by $X=X_E \in \mathbb{R}^{n \times q}$.
Clearly, the findings in our $K$ already existing studies provide a guideline for what covariates need to be included as possible confounders in the follow-up study.
We assume that the validation data is drawn from the same population model as \eqref{popn:modelk}, i.e., 
$$Y =  D \alpha + X_{E_0}\beta_{E_0} + \varepsilon,$$ 
where $\varepsilon$ is centered at $0_n$ and has covariance $\sigma^2 I$ given the treatment variable and covariates.

Fixing some more notation for the remaining paper, let $N_k= n+n_k$, and let $r_k= \frac{n_k}{N_k}$ in the rest of the paper.
In each existing study, we use the LASSO to estimate
\begin{equation}\label{eq:lasso}
\widehat{\gamma}^{\text{\;(L)}}_{k} = \begin{pmatrix} \widehat{\alpha}^{\text{\;(L)}}_{k} \\ \widehat{\beta}^{\text{\;(L)}}_{k}\end{pmatrix}=\underset{\alpha\in\mathbb{R}^s,\beta\in \mathbb{R}^{p_k} }{\arg\min}\left\{ \frac{1}{2r_k\sqrt{N_k}} ||Y_k -  D_k \alpha  -     X_k \beta  ||^2
+ || \Lambda_k \beta ||_1 \right\}, 
\end{equation}
$\widehat{\alpha}^{\text{\;(L)}}_{k}\in \mathbb{R}^s$, $\widehat{\beta}^{\text{\;(L)}}_{k} \in \mathbb{R}^{q_k}$ for $k=1,2,\cdots, K$, where $\Lambda_k \in R^{p_k \times p_k}$ is a diagonal matrix with the tuning parameters in the LASSO penalty in the diagonal entries. 
Denote by
\begin{equation}\label{eq:pooled-estimate:ls}
\hat{\gamma}_k =\begin{pmatrix}
\hat{\alpha}_k\\
\hat{\beta}_k
\end{pmatrix}
\end{equation} 
the least squares estimator, based on  
$\begin{bmatrix}X'_{k,E_k} & X'_{E_k}\end{bmatrix}'\in \mathbb{R}^{N_k \times q_k}$, \text{ and } $\begin{pmatrix} Y'_k & Y'\end{pmatrix}'\in \mathbb{R}^{N_k\times 1}$ after pooling observations from the validation study and existing study $k$.
\smallskip

\noindent{\textbf{Some extensions}}. 
The model in \eqref{popn:modelk} assumes that $\beta_{E_0}$, the effect of the true confounders, is the same across all the studies. 
Now we consider a more general setup with cross-study heterogeneity, and show that we can still work with our basic model with the same form as \eqref{popn:modelk}.
To see this, suppose that data in our existing studies are drawn from the  model
\begin{equation}
   Y_k =  D_k \alpha + X_{k,E_{0,k}}\beta_{E_{0,k}} + \varepsilon_k, 
   \label{popn:modelk:hetero}
\end{equation}
and data in our validation study are generated from the model
\begin{equation}
Y_0 =  D_0 \alpha + X_{0,E_{0,0}}\beta_{E_{0,0}} + \varepsilon_0,
\label{popn:model0:hetero}
\end{equation}
where $\mathbb{E}(\varepsilon_k | D_k, X_k) = 0_{n_k}$, $\text{Var}(\varepsilon_k | D_k, X_k) =\sigma^2 I$, for $k =0, \ldots, K$.
Let $\delta_{i,k}$ be a dummy variable that assumes the value $1$ if observation $i$ is in study $k$, and is $0$ otherwise.
Note that $\sum_{k=0}^K \delta_{i,k} =1$ for any observation $i$.
Thus, for a variable $v_i \in \mathbb{R}^t$, let
\[ \delta_{i,k}v_i = 
\begin{cases} 
      v_i & \text{ if  observation } i \in  \text{ study }k , \\
      0_t & \text{ otherwise}.
\end{cases}
\]

Given the above generating models, we can describe observation $i$ in our data, for the validation study as well as the $K$ existing studies, through the unified model
\begin{align*}
    y_i = &\ D'_i \alpha   +  Z'_{i, E_0}\gamma_{E_0} + \varepsilon_i,
\end{align*} 
where 
$$\gamma_{E_0}= \begin{pmatrix}  \beta_{E_{0,0}}' & \beta_{E_{0,1}}' &\cdots &\beta_{E_{0,K}}'  \end{pmatrix}', \text{ and } Z_{i, E_0} =\begin{bmatrix} (\delta_{i,0} X_{i, E_{0,0}})' & (\delta_{i,1} X_{i, E_{0,1}})' & \cdots & (\delta_{i,K} X_{i, E_{0,K}})' \end{bmatrix}'.$$
Because we can write the unified model in the same form as the model in \eqref{popn:modelk}, we stick to our basic model for ease of presentation and revisit the extensions after introducing our proposal.

\subsection{Debiasing through data carving: a first example}

We provide a first example to illustrate how we utilize data carving to construct an unbiased estimator of treatment effects.
Suppose that we have an existing study, i.e., $K=1$.
We let $\mathcal{D}_1$ contain a real valued covariate $X_1\in\mathbb{R}$ and a treatment variable $D_1$, i.e., $p_1=1$ and $s=1$. 
Both $X_1$ and $D_1$ in this example are standardized to have mean zero and variance one with correlation $\rho$.
Consider observing $n_1$ and $n$ observations in $\mathcal{D}_1 $ and $\mathcal{V}$ respectively, from the population model \eqref{popn:modelk} with $\beta_{E_0}=0$, i.e. $E_0 = \emptyset$, and with noise variance $\sigma^2=1$. 
For simplicity, we let $n=n_1$, which means that $N_1= 2n_1$, and $r_1 = \frac{1}{2}$ in this example. 
	 
Suppose that the model we select after solving the LASSO on $\mathcal{D}_1$ is the full model $E_1= \{1\}$ which includes the covariate $X_1$. 
That is, we model the real-valued outcome variable as independent realizations from the conditional distribution
\begin{equation}
	     \label{eg:working:model}
	     \mathrm{y}\ \lvert \ \mathrm{d}, \mathrm{x}_1 \sim N(\alpha \mathrm{d} + \beta \mathrm{x}_1,  1).
	 \end{equation}
We fit the selected model to the pooled data $\mathcal{D}_1 \cup \mathcal{V}$.
	 
Define $$\gamma_1 = \begin{pmatrix} \alpha & \beta \end{pmatrix}'.$$
Consistent with our notations, let $\hat\gamma_1$ be the least squares estimator for $\gamma_1$ using the pooled observations in $\mathcal{D}_1$ and $\mathcal{V}$.
Note that the likelihood in the selected model \eqref{eg:working:model} is
$$\ell_{N_1}(\gamma_1; \hat\gamma_1) \propto \exp\left(-\frac{N_1}{2}(\hat\gamma_1-\gamma_1)' \hat{\Sigma}_1 (\hat\gamma_1-\gamma_1)\right),$$  
if we ignored the impact of model selection bias.

To construct our carved estimator, we introduce a new variable which we call the randomization variable.
Denoted as
$$\omega_1= \begin{pmatrix} \omega_{1,1} & \omega_{1,2} \end{pmatrix}'\in \mathbb{R}^2,$$ 
the randomization variable that takes the value
\begin{equation*}
\begin{aligned}
	& \dfrac{\partial}{\partial\gamma_1}\Bigg(\dfrac{1}{2\sqrt{N_1}}\|Y_{1} - \alpha D_{1} - \beta X_{1}\|^2 + \dfrac{1}{2\sqrt{N_1}}\|Y - \alpha D - \beta X\|^2\\
	&\;\;\;\;\;\;\;\;\;\;\;\;\;\;\;\;\;\;\;\;\;\;\;\;\;\;\;\;\;\;\;\;\;\;\;\;\;\;\;\;\;\;\;\;\;\;\;\;\;\;\;\;\;\;\;\;\;\;\;\;\;\;\;-\dfrac{1}{2r_1\sqrt{N_1}}\|Y_{1} - \alpha D_{1} - \beta X_{1}\|^2\Bigg)\Bigg\lvert_{\hatG_1},
\end{aligned}
\end{equation*}
where $r_1=\frac{1}{2}$.
Observe that there is a simple equivalence between the Karush–Kuhn–Tucker (K.K.T.) mapping for \eqref{eq:lasso} and 
\begin{equation}
\label{lasso:randomized:rep}
     \underset{\alpha, \beta}{\text{minimize}} \;\; \dfrac{1}{2\sqrt{N_1}}\|Y -\alpha D  - \beta X \|^2  + \dfrac{1}{2\sqrt{N_1}}\|Y_1 -\alpha D_1 - \beta X_1 \|^2- \w_{1,1} \alpha -\w_{1,2} \beta + \lambda|\beta|,
\end{equation} 
which is a randomized version of the LASSO on the pooled data $\mathcal{D}_1 \cup \mathcal{V}$.
Under the selected model in \eqref{eg:working:model} and ignoring the impact of model selection, we have as $N_1\to \infty$: 
	 \begin{enumerate}
	 \setlength\itemsep{1em}
	     \item \label{prop:1} $\sqrt{N_1}(\widehat\gamma_1-\gamma_1)$ and $\omega_1$ jointly follow an Gaussian distribution such that $\omega_1$ is centered at $0$ with the covariance $\mathbb{E}[\hat\Sigma_1]$ 
	     \item \label{prop:2} $\omega_1$ is independent of $\sqrt{N_1}(\widehat\gamma_1-\gamma_1)\sim N_2(0_2, (\mathbb{E}[\widehat{\Sigma}_1])^{-1})$.
	 \end{enumerate}

Suppose for now, the limiting distributional properties listed under \ref{prop:1} and \ref{prop:2} hold exactly with $$\mathbb{E}[\hat\Sigma_1] = \begin{bmatrix} 1 & \rho \\ \rho & 1 \end{bmatrix}.$$ 
Our carved estimator can be constructed in two main steps.
We eliminate bias from model selection by conditioning on the selection event
\begin{equation}
	     \label{cond:event:1}
	\left\{ \widehat{E}_1 = E_1,  \text{sign}(\hatB_1) =s_{E_1} \right\}. 
	 \end{equation}

In Step 1, we recognize that
	$$
	    \hat\alpha_1^{\text{\;initial}}=\hat{\alpha}_1 - (1-\rho^2)^{-1} \dfrac{1}{\sqrt{N_1}}\cdot (\omega_{1,1}- \rho \omega_{1,2})
     $$
 is an unbiased, initial estimator that is independent of the selection through LASSO.
This leads us to observe that $\hat\alpha_1^{\text{\;initial}}$ is also unbiased for $\alpha_1$ conditional on the event in \eqref{cond:event:1}.

In Step 2, we improve upon the initial estimator through Rao-Blackwellization by conditioning further upon the complete sufficient statistic for $\gamma_1$. 
In our case, this statistic is $\widehat{\gamma}_1$, because of the basic fact that conditioning preserves the complete sufficient statistic, and that the sufficient statistic in the usual likelihood is the least squares estimator. 
Please see \citep{fithian2014optimal} for a formal proof of this fact. 
Thus, our carved estimator is given by
	 \begin{equation}
	 \begin{aligned}
	  \hat\alpha_1^{\text{\;carve}}
	  &= \mathbb{E}\Big[ \hat{\alpha}_1^{\text{initial}} \ \Big\lvert \hat{\gamma}_1, \widehat{E}_1 = E_1,  \text{sign}(\hatB_1) =s_{E_1}\Big],
	  \end{aligned}
	  \label{simple:carved}
	\end{equation}
	where we have conditioned upon $\hat{\gamma}_1$ alongside the event of selection in \eqref{cond:event:1}.
Details of both steps and a final expression for the carved estimator are provided in Proposition \ref{exact:UMVU:simple}.
Because 
$$\mathbb{E}\left[\widehat{\alpha}_1^{\text{\; carve}}\; \Big\lvert \; \widehat{E}_1 = E_1,  \text{sign}(\hatB_1) =s_{E_1}\right] = \alpha,$$
we also have $\mathbb{E}\left[\widehat{\alpha}_1^{\text{\; carve}}\right] = \alpha$.

A few remarks are in order.
First, the estimator in \eqref{simple:carved} can be written as
\begin{equation*}
\begin{aligned}
\hat\alpha_1^{\text{\;carve}} &= \widehat{\alpha}_1 - (1-\rho^2)^{-1} \dfrac{1}{\sqrt{N_1}} \cdot \mathbb{E}\Big[\omega_{1,1}- \rho \omega_{1,2} \;  \Big\lvert \;  \widehat{\gamma}_1,\widehat{E}_1 = E_1,  \text{sign}(\hatB_1) =s_{E_1}\Big],
\end{aligned}
\end{equation*}
In particular, we note that $\hat\alpha_1^{\text{\;carve}}$ applies an additive \textit{debiasing correction} to the simple least squares estimator $\widehat{\alpha}_1$.
The final expression for the carved estimator, as stated in Proposition \ref{exact:UMVU:simple}, is obtained by 
evaluating the conditional expectation in the debiasing correction.
Second, to construct the carved estimator, we used: (i) the least squares estimator in the selected model, (ii) the tuning parameter $\lambda$, and (iii) the sign of the LASSO estimator for the selected covariate $s_{E_1}$, which   depend on only some simple summary statistics from the existing study.

Next we extend our main ideas from this simple example to the more general setup.
Before proceeding, we provide some numerical evidence for the effectiveness of our new estimator.
We depict in Figure \ref{fig:boxplot-example} the result of a simulation for $n=100$,  $n_1 = 50$, and $1,000$ Monte Carlo samples.
We compare our carved estimator against the two popular alternatives that were discussed earlier in the introduction, (1) the estimator based on data splitting, which simply uses the validation data $\mathcal{V}$ for estimation, and (2) the debiased LASSO estimator based on a one-step bias correction to the LASSO estimator.
We elaborate on the generative scheme for the simulation as well as the implementation details for the other two alternatives in Section \ref{Sec:simulation}. 
Note that all the three estimators are centered around the true value $\alpha_1 = 1$. 
Especially, our new estimator removes the effect of selection bias, and has smallest variability among the three estimators.

	\begin{figure}
\centering        \includegraphics[width=0.50\textwidth]{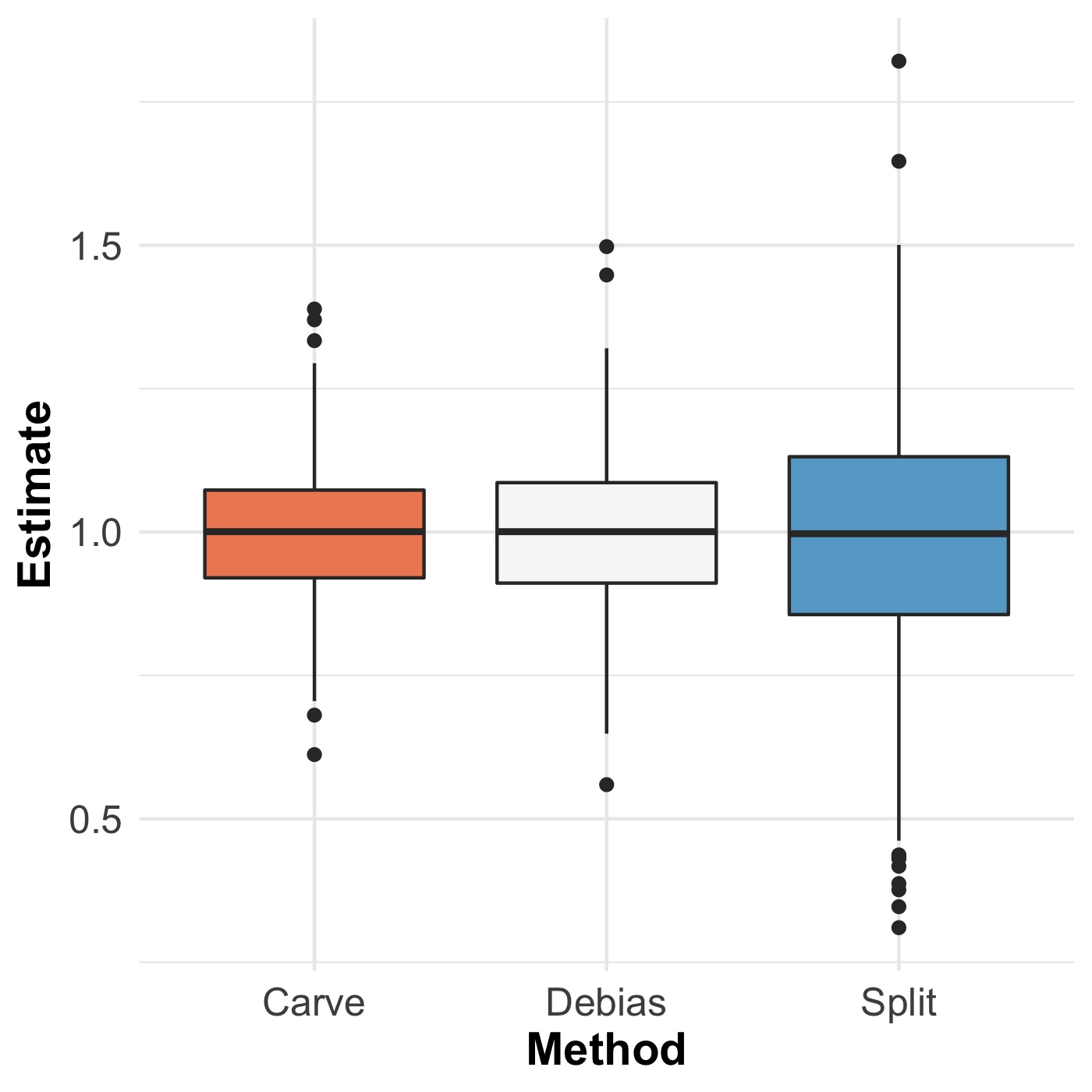}
        \caption{\label{fig:boxplot-example} Box-plots of the parameter estimates in the simulation experiment for Section \ref{Sec:framework}.}
    \end{figure}

\section{Carved estimator} 
\label{Sec:theory}

\subsection{Privacy-preserving aggregation}

We turn to our general setup, and outline our scheme to aggregate summary statistics from existing studies with data carving. 

We begin by specifying the summary statistics involved in the construct of our estimator:
\begin{enumerate}
\setlength\itemsep{1em}
	\item \textit{Summary from model selection in existing studies}: the support set of the LASSO estimator $E_k$, the penalty weights $\Lambda_{ E_k}$, and the signs of the (nonzero) LASSO estimator for the selected covariates $s_{E_k} = \text{sign}(\widehat{\beta}_k^{\text{\;(L)}})$. \label{summary:sel}
	\item \textit{Summary based on the first two moments from data in existing studies}: the sample covariance matrices from the existing study $k$, based on the selected model $E_k$ and the observed sample of size $n_k$, which are equal to
	\begin{align*}
	\widehat{\xi}_k = \left[ \begin{array}{cc}
	\hat{\xi}_{k, 01} \\
	\hat{\xi}_{k, 02} 
	\end{array} \right]   =&  \frac{1}{n_k} \left[
	\begin{array}{c}
	D_{k}'Y_k \\
	X_{k, E_k}'Y_k 
	\end{array}  \right], \\
	\hat{\Xi}_{k} =  \ \left[
	\begin{array}{cc}
	\hat{\Xi}_{k, 11} & \hat{\Xi}_{k, 12}\\
	\hat{\Xi}_{k, 21} & \hat{\Xi}_{k, 22} 
	\end{array} \right]   =  & \frac{1}{n_k} \left[
	\begin{array}{cc}
	D_k'D_k & D_k'X_{k, E_k}\\
	X_{k, E_k}'D_k & X_{k,  E_k}'X_{k, E_k}
	\end{array} \right].
	\end{align*} 
	Note that $\widehat{\xi}_k$ has two blocks, one $s$-dimensional and the other  $q_k$-dimensional. In a similar fashion,  $\hat{\Xi}_{k} $ is partitioned along the same dimensional structure. 
	\label{summary:data}
\end{enumerate}

First, we compute the least squares estimator 
\begin{align}\label{eq:pooled-estimate}
\hat{\gamma}_k =\begin{pmatrix}
\hat{\alpha}_k\\
\hat{\beta}_k
\end{pmatrix} = \begin{pmatrix}
n_k\hat{\Xi}_{k, 11} + D'D &  n_k\hat{\Xi}_{k, 12} + D'X_{E_k}\\ 
n_k\hat{\Xi}_{k, 12} +X'_{E_k}D &  n_k\hat{\Xi}_{k, 22}  + X_{E_k}'X_{E_k}
\end{pmatrix}^{-1}  \begin{pmatrix}
n_k  \hat{\xi}_{k, 01}  + D'Y\\
n_k  \hat{\xi}_{k, 02} +   X_{E_k}'Y
\end{pmatrix}.
\end{align}

For a fixed vector $u\in \mathbb{R}^k$ and for $v\in \mathbb{R}^k$, let $\log(1+ \frac{1}{uv})$ be a $\mathbb{R}^k$-valued vector whose $j$-th component equals $$\log\left(1+ \frac{1}{u_j v_j}\right).$$
Denote by 
\begin{equation}
\hat{\Sigma}_{k} = \frac{1}{N_k}\left(n_k \hat{\Xi}_{k}  + \begin{pmatrix} D & X_{E_k} \end{pmatrix}' \begin{pmatrix} D & X_{E_k} \end{pmatrix}\right),
\label{sample:pooled:covariance}
\end{equation}
the sample covariance matrix based on the selected covariates. 
We  solve $\hat{z}_k=\begin{pmatrix} \hat{z}'_{k,1} & \hat{z}'_{k,2} \end{pmatrix}'$ from the following convex optimization problem: 
\begin{equation}
\label{optimizer}
\begin{aligned}
& \underset{z_k\in\mathbb{R}^{s+q_k}}{\text{arginf}} \Bigg\{ (1-r_k)^{-1}r_k\frac{1}{2}\left(\sqrt{N_k}z_k-\sqrt{N_k}\widehat{\gamma}_k +  \hat{\Sigma}_{k}^{-1}\begin{pmatrix} 0' & (\Lambda_{E_k} s_{E_k})'\end{pmatrix}\right)'\hat{\Sigma}_{k} \\
&\;\;\;\;\;\;\;\;\;\;\;\;\;\;\;\;\;\;\;\;\;\;\;\;\;\;\;\;\;\;\; \left(\sqrt{N_k}z_k-\sqrt{N_k}\widehat{\gamma}_k +  \hat{\Sigma}_{k}^{-1}\begin{pmatrix} 0' & (\Lambda_{E_k} s_{E_k})'\end{pmatrix}\right)+ B_{s_{E_k}}\left(\sqrt{N_k}z_{k,2}\right)\Bigg\},
\end{aligned}
\end{equation}
where    $$B_{s_{E_k}}\left(\sqrt{N_k}z_{k,2}\right)= 1_k'\log\left(1+ \frac{1}{s_{E_k} \sqrt{N_k}z_{k,2}}\right).$$
Our carved estimator pools data from the existing study $k$ with the validation study, and is given by
\begin{align}\label{eq:gamma-carve}
\widehat{\alpha}_k^{\text{\; carve}} =\widehat{\alpha}_{k} + (1-  r_{k})^{-1}r_k \left(\widehat{\alpha}_k -  
\frac{1}{\sqrt{N_k}}e'_{1}\hat{\Sigma}_{k}^{-1}\begin{pmatrix} 0' & (\Lambda_{E_k} s_{E_k})'\end{pmatrix}' - e'_{1}\hat{z}_k\right). 
\end{align}

Finally, we propose the aggregated estimator by taking a simple average of the $K$ carved estimators as
\begin{align}\label{proposed estimate}
\tilde{\alpha}  = \frac{1}{K}\sum_{k=1}^{K} \widehat{\alpha}_k^{\text{\; carve}} . 
\end{align}

We conclude the section with some observations about our estimator.
\begin{remark}
If the sample size for each existing study, $n_k$, varies with $k$, we can replace the simple average in (\ref{proposed estimate}) by a weighted average, where the weights are proportional to $N_k$. 
\end{remark}

\begin{remark}
    Let
  $$Q= \max_{k\in \{1,\ldots, K\}} |E_k|$$
  denote the maximum size (cardinality) for the selected covariate sets across our existing studies. 
  Then, we note that the amount of information required for constructing the estimator in \eqref{proposed estimate} is only $O(Q^2)$.
\end{remark}

\begin{remark}
We revisit the population models in \eqref{popn:modelk:hetero} and \eqref{popn:model0:hetero} where we considered cross-study heterogeneity.
We observe that our estimator easily adapts to the following two scenarios. 
First, say $E_{0,0} \subseteq \cap_{k=1}^K E_{0,k}$.
In this case, we can proceed as prescribed above.
Second, suppose that $E_{0,0} \subseteq \cup_{k=1}^K E_{0,k}$. In this case, we assume that the union of the selected variables $E = \cup_{k=1}^K E_k$ is collected in all the existing studies.
We can now form our carved estimator by fitting the union model $E$ to the pooled data from existing study $k$ and the validation study.
\end{remark}

\subsection{Our debiasing approach}

In this section, we show how the debiasing approach in the first example can be generalized to our problem.
In line with the first example, our carved estimator $\widehat{\alpha}_k^{\text{\; carve}}$ is obtained by conditioning on the selection event
\begin{equation}
\label{sel:event:lasso}
\left\{\widehat{E}_k =  E_k, \;\text{sign}(\hatB_k) = s_{E_k}\right\},
\end{equation}
which we observe after solving \eqref{eq:lasso} on data $\dcal_k$.
As before, we work with the selected model which assumes that data for the outcome variable are drawn as independent realizations from the conditional distribution
$$\mathrm{y}\ \lvert \ \mathrm{d}, \mathrm{x}_{E_k} \sim N(\mathrm{d}'\alpha + \mathrm{x}_{E_k}'\beta_{E_k}, \sigma^2).$$ 
We proceed by fitting the selected model to $\dcal_k \cup \mathcal{V}$.
Since we work under a fixed $\sigma^2$ setting, we will assume $\sigma^2=1$ for the sake of simplicity.

To develop the theory, we let
\begin{equation}
\begin{aligned}
\omega_k &= \begin{pmatrix} (\omega_{k,1})' & (\omega_{k,2})' \end{pmatrix}'; \; \omega_{k,1}\in \mathbb{R}^s, \omega_{k,2}\in \mathbb{R}^{p_k}\\
&=\dfrac{\partial}{\partial\gamma}\Bigg(\dfrac{1}{2\sqrt{N_k}}\|Y_{k} -D_{k} \alpha  -  X_{k}\beta\|^2 + \dfrac{1}{2\sqrt{N_k}}\|Y -  D\alpha -  X\beta\|^2\\
&\;\;\;\;\;\;\;\;\;\;\;\;\;\;\;\;\;\;-\dfrac{1}{2r_k\sqrt{N_k}}\|Y_{k} - D_{k}\alpha -  X_{k}\beta\|^2\Bigg)\Bigg\lvert_{\hatG_k},
\label{randomization:gen}
\end{aligned}
\end{equation}
be our randomization variable, which is based on the LASSO solution $\hatG_k$.
Let $E_k^c$ denote indices of covariates that are not selected by the LASSO in existing study $k$.
In addition, we will also consider the statistic
$$\hat{\Gamma}_{k} = \dfrac{1}{N_k}\left\{X'_{k,E_k^c}(Y_k-D_k\hat\alpha_k- X_{k, E_k}\hat\beta_k) + X'_{E_k^c}(Y-D\hat\alpha_k- X_{E_k}\hat\beta_k) \right\}.$$
Note, we need not observe the covariates that were not selected by the LASSO in the validation study.
The statistics defined above only serve to provide a theoretical justification for our debiasing term.
Proofs of our technical results are provided in Appendix \ref{Sec:proofs}.

Lemma \ref{limiting:Gaussianity} notes that the variables $\widehat\gamma_k$, $\Gamma_k$ and
$\omega_k$ admit an asymptotic linear representation in the selected model, if we ignored the impact of model selection and treated $E_k$ as a fixed set. Formally, this implies that these variables are distributed as Gaussian variables in the limit as the sample size $N_k\to \infty$.

\begin{lemma}
	\label{limiting:Gaussianity}
	Let $\gamma_k =\begin{pmatrix} \alpha' & \beta'_{E_k}\end{pmatrix}'$ for a fixed subset $E_k$.
	Then, the following assertion holds:
	$$\begin{pmatrix} \sqrt{N_k}(\widehat\gamma_k-\gamma_k)' & \sqrt{N_k}\widehat\Gamma'_k & \omega'_k\end{pmatrix}'= \sqrt{N_k}\bar{T}_{N_k} + R_{N_k}, $$
	where (i) $\bar{T}_{N_k}$ is the average of $N_k$ i.i.d. variables with mean equal to $0_{2p_k+ 2s_k}$, (ii) $R_{N_k}= o_p(1)$, (iii)  $\hat{\gamma}_k$, $\hat{\Gamma}_{k}$ and $\omega_k$ are asymptotically independent, and (iv) $\omega_k$ is asymptotically centered at $0$ with the covariance
	$(1-r_k)r_k^{-1} \mathbb{E}[\mathfrak{G}_k]$, where $\mathfrak{G}_k =  {N_k}^{-1}\left\{\begin{pmatrix} D & X  \end{pmatrix}' \begin{pmatrix} D & X  \end{pmatrix} + \begin{pmatrix} D_k & X_k  \end{pmatrix}'  \begin{pmatrix} D_k & X_k  \end{pmatrix}\right\}.$
\end{lemma}

Suppose that we condition this limiting Gaussian distribution on the selection event in \eqref{sel:event:lasso}.
This gives us an asymptotic conditional distribution.
Say that we denote this distribution by $\widetilde{L}_{N_k}$.
Next, Proposition \ref{exact:UMVU} derives a debiasing correction from this asymptotic conditional distribution as was done in the first example. 
Before stating this result, we let
$$\mu_{\mathcal{H}_0}(\alpha_0, \Sigma_0)$$
denote the first moment of a Gaussian variable with mean $\alpha_0$ and covariance $\Sigma_0$ that is truncated to the region $\mathcal{H}_0$. Further, let $\Delta(\eta_0)$ be the log-partition function for the related truncated Gaussian density at the natural parameter $$\eta_0= \Sigma_0^{-1}\alpha_0.$$
Finally, for $\widehat{\Sigma}_k$ defined in \eqref{sample:pooled:covariance}, let $\Sigma_k =\mathbb{E}[\widehat{\Sigma}_k]$ be the population covariance.

\begin{proposition}
	\label{exact:UMVU}
	Define the half-space
	$$\mathcal{H} =\left\{z=(z_1, z_2), \; z_1 \in \mathbb{R}^{s}, \; z_2 \in \mathbb{R}^{q_k}: \text{sign}(z_{2}) = s_{E_k}\right\}.$$
	Then, the estimator 
	\begin{equation*}
	\hat\alpha_k + (1-r_k)^{-1}{r_k}\left(\hat\alpha_k-\frac{1}{ \sqrt{N_k}}\cdot e'_1 {\Sigma}^{-1}_{k} \zeta_k  - \frac{1}{ \sqrt{N_k}}\cdot e'_1\mu_{\mathcal{H}}\left(\sqrt{N_k}\hat\gamma_k -\Sigma^{-1}_{k} \zeta_k, r_k^{-1}(1-r_k){\Sigma}^{-1}_k\right)\right)
	\end{equation*}
	is unbiased for $\alpha$ with respect to $\widetilde{L}_{N_k}$.
\end{proposition}

The estimator in Theorem \ref{exact:UMVU}, however, does not admit direct computations due to lack of readily available expressions for the moments of a truncated Gaussian variable. 
To this end, we use the plug-in estimator $\widehat{\Sigma}_k$ and the inverse of $\widehat{\Sigma}_k$ for estimating $\Sigma_k$ and its inverse, respectively, and apply a Laplace-type approximation to facilitate a manageable calculation for this estimator.
Deferring a rigorous justification of asymptotic unbiasedness to the next discussion, we complete the steps to numerically approximate our estimator through an easy-to-solve optimization.
For $$\eta_0=(1-r_k)^{-1}r_k\widehat{\Sigma}_k(\sqrt{N_k}\hat\gamma_k -\widehat\Sigma^{-1}_{k} \zeta_k),$$
a Laplace approximation grants us the following
\begin{equation*}
\begin{aligned}
& \Delta(\eta_0) \approx  (1-r_k)^{-1}r_k \dfrac{1}{2} (\sqrt{N_k}\hat\gamma_k - \widehat\Sigma^{-1}_{k} \zeta_k)' \widehat{\Sigma}_{k} (\sqrt{N_k}\hat\gamma_k - \widehat\Sigma^{-1}_{k} \zeta_k) \\
&-\displaystyle\inf_{\text{sign}(z_{k,2}) = s_{E_k}}\;\; (1-r_k)^{-1}r_k \dfrac{1}{2} (\sqrt{N_k}z_{k}- \sqrt{N_k}\hat\gamma_k + \widehat\Sigma^{-1}_{k} \zeta_k)'  \widehat{\Sigma}_{k} (\sqrt{N_k}z_{k}- \sqrt{N_k}\hat\gamma_k + \widehat\Sigma^{-1}_{k}
\zeta_k)\\
&=\displaystyle\sup_{\text{sign}(z_{k,2}) = s_{E_k}}\;\; (1-r_k)^{-1}r_k (\sqrt{N_k}z_{k})' \widehat{\Sigma}_{k} (\sqrt{N_k}\hat\gamma_k - \widehat\Sigma^{-1}_{k} \zeta_k)\\
&\;\;\;\;\;\;\;\;\;\;\;\;\;\;\;\;\;\;\;\;\;\;\;\;\;\;\;\;\;\;\;\;\;\;\;\;\;\;\;\;\;\;\;\;\;\;\;\;\;\;\;\;\;\;\;\;\;\;\;\;\;\;\;\;\;\;\;\;\;\;\;\;\;\;\;\;\;\;\;\;\;\;\;\;\;\;\;\;\;\;\;\;\;\;\;\;\;\;\;\;\;\;\;\;\;\;\;\;- (1-r_k)^{-1}r_k \dfrac{1}{2} (\sqrt{N_k}z_{k})'\widehat{\Sigma}_{k} \sqrt{N_k} z_{k}.
\end{aligned}
\end{equation*}
We replace the constrained optimization with the unconstrained version through a logarithmic barrier penalty to obtain
\begin{equation*}
\begin{aligned}
\Delta(\eta_0) &\approx \displaystyle\sup_{z_k} \;\; (1-r_k)^{-1}r_k (\sqrt{N_k}z_{k})' \widehat{\Sigma}_{k} (\sqrt{N_k}\hat\gamma_k - \widehat\Sigma^{-1}_{k} \zeta_k) \\
&\;\;\;\;\;\;\;\;\;\;\;\;\;\;\;\;\;\;\;\;\;\;\;\;\;\;\;\;\;\;\;\;\;\;\;\;\;\;\;\;\;\;\;\;\;\;\;\;\;\;\;\;\;- (1-r_k)^{-1}r_k \dfrac{1}{2} (\sqrt{N_k}z_{k})'\widehat{\Sigma}_{k} \sqrt{N_k}z_{k}- B_{s_{E_k}}(\sqrt{N_k} z_{k,2}).
\end{aligned}
\end{equation*}
Applying the above approximation to our debiasing correction in Proposition \ref{exact:UMVU} allows us to write
\begin{equation*}
\begin{aligned}
\mu_{\mathcal{H}}\left(\sqrt{N_k}\hat\gamma_k -\widehat\Sigma^{-1}_{k} \zeta_k, r_k^{-1}(1-r_k)\widehat{\Sigma}^{-1}_k\right) &=\nabla \Delta( (1-r_k)^{-1}r_k \widehat{\Sigma}_{k} (\sqrt{N_k}\hat\gamma_k - \widehat\Sigma^{-1}_{k} \zeta_k))\\
& \approx \sqrt{N}_k\widehat{z}_k,
\end{aligned}
\end{equation*}
where $\widehat{z}_k$ is the solution in \eqref{optimizer}.
This gives us the expression of our carved estimator in \eqref{eq:gamma-carve}.

\subsection{Asymptotic properties}

Our main result in the section, Theorem \ref{unbiasedness}, establishes the asymptotic unbiasedness of our carved estimator.

Suppose that we have $N$ i.i.d. realizations of the response $\mathrm{y}\in \mathbb{R}$, the covariates for all existing studies $\mathrm{x}\in \mathbb{R}^{p_1+ \cdots+p_k}$ and the treatment $\mathrm{d}\in \mathbb{R}^{s}$, such that
 $\mathrm{y}$ is drawn from the model in \eqref{popn:modelk} with  
 $$\begin{pmatrix} \alpha' & \beta'_{E_0}\end{pmatrix}'=\begin{pmatrix} (\alpha^{(N)})' & (\beta^{(N)}_{E_0})'\end{pmatrix}'.$$ 
Each of our existing studies, in particular, observes $N_k$ samples for the outcome and treatment along with a subset of covariates that are seen across all the studies. 
For the remaining section, we consider the following parameters
\begin{equation}
    \label{seq:parameters}
    \sqrt{N}\begin{pmatrix} (\alpha^{(N)})' & (\beta^{(N)}_{E_0})'\end{pmatrix}' = a_N \begin{pmatrix} (\alpha_0)' & (\beta_{0,E_0})'\end{pmatrix} 
\end{equation}
in our generation scheme, such that $\alpha_0$ and $\beta_{0,E_0}$ are constants and $a_N= o(\sqrt{N})$.
We state the regularity and moment conditions on the data-generating mechanism below. 

To start with, we impose asymptotic screening as a requirement on the model selection in the existing studies.
This requirement is necessary to ensure that the bias from missing a true covariate is only negligible with growing sample size.
The LASSO is an important example that satisfies the asymptotic screening requirement under conditions discussed at length by \cite{meinshausen2006high, bickel2009simultaneous}.
Similar conditions are required by estimators which are aggregated over selected models, see for example the work by \cite{schultheiss2021multicarving, fei2021estimation}.
\begin{condition}
    Assume that $\displaystyle\lim_{n_{k}\rightarrow \infty}\mathbb{P}(E_0\subseteq E_k)\rightarrow 1 $, for $k=1,2, \ldots,K$. 
    \label{screening:condition}
\end{condition}

The second condition below allows us to control the bias from the use of a Laplace-type approximation for our debiasing correction.
The first part of the condition assumes the existence of an exponential moment in a neighborhood of zero. 
This is required to justify a Laplace-type approximation for probabilities involving the mean of $N$ i.i.d. variables, which in our case is $\bar{T}_{N_k}$ (as stated in Lemma \ref{limiting:Gaussianity}) by setting $N= N_k$.
The second part of the condition is required to extend this approximation to asymptotically linear variables with an added $o_p(1)$ error term, which we see in the left-hand side of the representation in Lemma \ref{limiting:Gaussianity}.
\begin{condition}
\label{moment:condition:0}
Fix $$\mathrm{V}= ( \mathrm{y} - \alpha'\mathrm{d}  -\beta'_{E_0}\mathrm{x}_{E_0})\cdot \begin{pmatrix} \mathrm{d}' & \mathrm{x}'\end{pmatrix}',$$
for a fixed subcollection of our covariates containing $E_0$.
We assume that
$$\mathbb{E}\left[\exp\left(\eta \|\mathrm{V}\|\right)\right] <\infty$$
for some $\eta\in \mathbb{R}^{+}$.
For the observed set of covariates $E_k$ in the existing study k, consider the linearizable representation in Lemma \ref{limiting:Gaussianity}. 
We assume that the remainder term $R_{N_k}$ satisfies:
\begin{equation*} 
\displaystyle\lim_{N_k\to \infty}  \dfrac{1}{a_{N_k}^2}\cdot \log \mathbb{P}\left[a_{N_k}^{-1}\| R_{N_k} \| > \epsilon \right] =- \infty \ \ \ \ \ \ \ \text{ for every } \epsilon >0.
 \end{equation*}
\end{condition}

The probability of selection can be written as the probability of a polyhedron:
$$
\mathbb{P}\left[\widehat{E}_k =  E_k, \;\text{sign}(\hatB_k) = s_{E_k}\right]= \mathbb{P}\left[A\begin{pmatrix} \sqrt{N_k}\widehat\gamma'_k & \sqrt{N_k}\widehat\Gamma'_k & \omega'_k\end{pmatrix}' + o_p(1) \leq b\right],
$$
where $A$ and $b$ are fixed matrices. The above   characterization for the selection event has been proved in Proposition 4.2 of \cite{panigrahi2016integrative}.
Our next condition allows us to replace the log-probability of selection with the log-probability of the polyhedron that is determined by matrices $A$ and $b$.
In other words, this condition allows us to ignore the $o_p(1)$ remainder in the characterization for our selection event.
\begin{condition}
\label{prob:sel:condition:0}
We assume that the probability of selection satisfies
$$\lim\dfrac{1}{a_{N_k}^2}\left\{\log \mathbb{P}\left[\widehat{E}_k =  E_k, \;\text{sign}(\hatB_k) = s_{E_k}\right]-\log \mathbb{P}\left[A\begin{pmatrix} \sqrt{N_k}\widehat\gamma'_k & \sqrt{N_k}\widehat\Gamma'_k & \omega'_k\end{pmatrix}' \leq b\right] \right\}=0.$$
\end{condition}

The final condition bounds the deviation of the sample estimate $\widehat{\Sigma}_k$ from the population parameter $\mathbb{E}[\widehat{\Sigma}_k]$.
\begin{condition}
For the selected set $E_k$, we assume that
$\mathbb{E}[\|\widehat\Sigma_{k}-\Sigma_k\|_{\text{op}}] = O(1)$, where $\| M \|_{\text{op}}$ denotes the operator norm of the matrix $M$.
\label{operator:norm:bdd:0}
\end{condition}

For the remaining section, consider the parameters in \eqref{seq:parameters} by fixing $N= N_k$.
\begin{theorem}
	\label{unbiasedness}
	Let $\widehat{\alpha}_k^{\text{\; carve}}$ be defined according to \eqref{eq:gamma-carve}. 
	Under the conditions in \ref{screening:condition}, \ref{moment:condition:0}, \ref{prob:sel:condition:0} and \ref{operator:norm:bdd:0}, 
	we have $$\mathbb{E}\left[ \|\sqrt{N_k}(\widehat{\alpha}_k^{\text{\; carve}} - \alpha) \|^2  \right]\leq  a_{N_k}^{2} C,$$ where $C$ is a constant.
\end{theorem}

The asymptotic variance of our estimator $\sqrt{N_k}\widehat{\alpha}_k^{\text{\; carve}}$, based on the observed Fisher information matrix, is
$$\widehat{V}_{E_k}^{\; \text{carve}}= e_1'\left((1-r_k)^{-1}\hat{\Sigma}^{-1}_{k}  - (1-r_k)^{-2}r^2_k\left((1-r_k)^{-1}r_k\hat{\Sigma}_{k} + \nabla^2 B_{s_{E_k}}(\sqrt{N_k} \widehat{z}_k)\right)^{-1}\right)e_1.$$
The expression for the estimated variance is detailed out in Proposition \ref{Fisher:info}, under Section \ref{Sec:proofs}.
Lemma \ref{efficiency:gain:K1} then quantifies the relative gain in variance of our estimator over splitting for each study $k$. 
Here, $\widehat{V}^{\; \text{split}}_{E_k}$ denotes the estimated variance of the least squares estimate when refitted on the validation data alone.
Let $\widehat{V}_{j,E_k}^{\; \text{carve}}$ and $\widehat{V}^{\; \text{split}}_{j,E_k}$ be the $j$-th diagonal entry of the corresponding $s$-dimensional matrices, respectively. 

\begin{lemma}
	\label{efficiency:gain:K1}
	Let $r_k$ be the ratio $\frac{n_k}{N_k}$, and let $B_{\text{max}}$ be the maximum value of the $(\mathbb{R}^+)^{s+q_k}$-valued vector
	$$(s_{E_k} \sqrt{N_k}\widehat{z}_k)^{-2} - (1+s_{E_k} \sqrt{N_k}\widehat{z}_k)^{-2},$$ and let $\lambda_{\text{min}}$ be the smallest eigen value of $\widehat{\Sigma}_k$. 
	Then, the following holds for $j\in \{1,\ldots, K\}$:
	$$\left(\widehat{V}^{\; \text{split}}_{j, E_k}\right)^{-1} (\widehat{V}^{\; \text{split}}_{j, E_k} - \widehat{V}^{\; \text{carve}}_{j, E_k}) \geq (1-r_k)^{-1} r^2_k \left((1-r_k)^{-1} r_k+  B_{\text{max}} \lambda_{\text{min}}^{-1} \right)^{-1}.$$
\end{lemma}

Clearly, our carved estimator $\widehat{\alpha}_k^{\text{\; carve}}$ dominates the split estimators in variance. 
The variance of the averaged carved estimator, however, involves correlations between the refitted estimators from the selected models on any pair of existing studies in addition to the variances from each study. 
Heuristically, the correlation between a pair of split estimators, using the validation data alone, is expected to be larger than the carved counterpart.
This is because the carved estimator also uses  information from the independent samples in this pair of existing studies.
Consider, for instance, the situation when each study reports the same model
$$E_1=E_2=\cdots= E_k.$$
Let the bound on the right-hand side of Lemma \ref{efficiency:gain:K1} be equal to $B_k$. 
Comparing the averaged carved estimator in \eqref{proposed estimate} with the split estimator, the difference in the two variances is seen to be bounded below by
$$
\frac{1}{K^2}\left(\sum_{k=1}^K \frac{B_k}{1-B_k} \widehat{V}_{E_k}^{\; \text{carve}} +  2\sum_{k_1< k_2} \left(\frac{1}{\sqrt{(1-B_{k_1})(1-B_{k_2})}} -1\right) \sqrt{\widehat{V}_{E_{k_1}}^{\; \text{carve}}}\sqrt{\widehat{V}_{E_{k_2}}^{\; \text{carve}}}\right).
$$
Note, this bound is strictly greater than $0$ since $B_k<r_k<1$ for $k\in 1,\ldots,K$.

Our simulation studies in the next section empirically support the gain in efficiency that we claim for our carved estimator.

\begin{remark}
In this section, we introduced a new way to remove bias from the refitted least squares estimator in a selected model and aggregate summary statistics from existing studies for unbiased estimation of treatment effects.
More generally, our new estimator extends to the  class of generalized linear models (GLMs).
In Appendix \ref{Appendix:A} ,we develop an extension of our proposal to the widely used logistic regression model, which is a concrete instance in the class of GLMs.  
\end{remark}

\section{Simulation studies}\label{Sec:simulation}

We conduct simulation studies to evaluate the finite-sample performance of our new estimator.
The outcome variables, in the existing and validation studies, are generated through simple linear models as follows
\begin{align}\label{eq:Simulation-no-D-model}
&\nonumber Y_k = \alpha D_k+ X_{k, E_0}\beta_{E_0} + \varepsilon_k, \quad \varepsilon_k \sim N(0, \sigma^2_{\varepsilon} I_{n_k}), \quad \text{for }k=1,2, \ldots,K,\\
& Y = \alpha D+ X_{E_0}\beta_{E_0} + \varepsilon, \quad \varepsilon \sim N(0, \sigma^2_{\varepsilon}I_{n}).
\end{align}
We set the coefficients $\alpha = 1$ and $\beta_{E_0} = (1.5, 1, 1, 1, 1)'$, i.e., for each study $k$, $X_{k, E_0}$ is the first 4 columns of $X_k$. 
The dimension of the covariates varies with $k$ as $p_k = 400 + 20k$ and $p=500$. 
Our sample sizes for the existing studies are $n_1 = \ldots = n_K= 100$, where $K \in \{ 2,3,5,10\} $; the validation study has $n=50$ samples. 
Our data is generated under two signal-to-noise-ratio values by using two values of  $\sigma_{\varepsilon} \in \{ 2,4\}$.

We summarize the performance of our carved estimator in three different settings.
\begin{enumerate}
\setlength\itemsep{1em}
    \item \label{setting:1} Setting I. \; In the first setting, there is a moderate degree of correlation between the treatment assignment and the covariates. We generate $(D_k, X_k)$ and $(D, X)$ from a multivariate normal distribution $ N(0,\Sigma)$ where $\Sigma = \big( \Sigma_{jl}\big)_{j,l=1}^{p_k+1}$ and $\Sigma_{jl} = 0.5^{|j-l|}$. 
    \item \label{setting:2} Setting II. \; The second setting generates highly correlated treatment and covariates  by first generating 
    \begin{align}\label{eq:Simulation-D-model}
     & D = X_{ M_0}'\gamma_{M_0} + \nu,\ \nu \sim N(0, \sigma_{\nu}^2I_{n}), \quad  D_k = X_{k, M_0}'\gamma_{M_0} + \nu_k,\ \nu_k \sim N(0, \sigma_{\nu}^2I_{n_k}),
     \end{align} 
     followed by drawing the response according to \eqref{eq:Simulation-no-D-model}.
     In the model \eqref{eq:Simulation-D-model} for generating the treatment variable, we fix $M_0 = \{5,6\}$, $\gamma_{M_0} = (1,1)'$, and fix the noise variance $\sigma_{\nu}^2 = 0.1$. 
     \item \label{setting:3} Setting III. \; In the final setting, we vary the coefficients of our covariates for generating the data in the existing and validation studies. 
     We simulate the shift in covariates across the existing and validation data as follows.
     We draw our data according to the generating scheme described for Setting II, except now we set the value of covariate coefficients $\beta_{E_0} = (1.5, 1, 1, 1, 0)'$ in model \eqref{eq:Simulation-no-D-model} to draw the response in our validation study. That is, the coefficient of one of the confounders has changed between the existing studies and the validation study in the follow-up stage.
\end{enumerate}

We compare the performance of our carved estimator with the split-and-aggregate estimate and the aggregated debiased LASSO estimator.
As remarked early on, an off-the-shelf option available for estimation is a simple split-and-aggregate estimator.
This estimator is obtained by refitting the selected model $E_k$ on the validation data $\mathcal{V}$ and then averaging the resulting $K$ split estimators.
Another legitimate estimator is the aggregated debiased LASSO estimator that can be viewed as the proposal in \cite{cai2021individual} when customized to our problem setup. 
To be specific, the aggregated estimator is formed by averaging the $K+1$ debiased LASSO estimators from the existing studies $\mathcal{D}_k$ and the validation dataset $\mathcal{V}$. 
Given that the debiased LASSO and the post-double selection estimates are asymptotically equivalent (see for example \citep{wang2020debiased}), we implement the debiased LASSO estimator via the  \texttt{R} package \texttt{hdm} for its computational speed. We found that in a small number of cases, the package \texttt{hdm} produced  numerically unstable results, and therefore, we report only the median bias and the median squared errors, instead of the averages that would be heavily against the debiased LASSO estimator due to a few unstable cases.  
We report both the mean and median bias and squared errors for our carved estimator and the split-and-aggregate estimator.
The comparison of our estimator with debiased LASSO should be made with respect to the corresponding medians, even though we provided mean for the carved and the split estimators for a relative comparison between these two estimators.

\begin{table}[h!]
\def~{\hphantom{0}}
\centering
\resizebox{\columnwidth}{!}{\begin{tabular}{cccccccccccccc}
        & &  \multicolumn{5}{c}{Bias} &   \multicolumn{5}{c}{MSE}   \\
    $\sigma_{\varepsilon} $   & $K$ & \multicolumn{2}{c}{Carved} & \multicolumn{2}{c}{Split} & Debiased  &  \multicolumn{2}{c}{Carved} & \multicolumn{2}{c}{Split} & Debiased \\
4 & 3 & $0.107$ & $0.095$ & $0.119$ & $0.103$ & $-0.051$ &   $0.109$ & $0.042$ & $0.248$ & $0.103$ & $0.044$ \\
4 & 5 & $0.117$ & $0.122$ & $0.123$ & $0.109$ & $-0.028$ &   $0.065$ & $0.032$ & $0.157$ & $0.077$ & $0.024$   \\
4 & 10 & $0.127$ & $0.125$ & $0.124$ & $0.124$ & $-0.046$ &   $0.043$ & $0.021$ & $0.083$ & $0.043$ & $0.014$  \\
2 & 3 & $-0.003$ & $0.005$ & $0.010$ & $-0.001$ & $-0.076$ &   $0.017$ & $0.006$ & $0.042$ & $0.017$ & $0.013$   \\
2 & 5 & $0.002$ & $-0.002$ & $0.012$ & $0.009$ & $-0.061$ &   $0.010$ & $0.004$ & $0.024$ & $0.009$ & $0.008$  \\
2 & 10 & $0.004$ & $0.003$ & $0.010$ & $0.011$ & $-0.051$ &   $0.004$ & $0.002$ & $0.012$ & $0.006$ & $0.005$   \\
    \end{tabular}}
    \vspace{0.2cm}
\caption{\normalfont{Simulation results under Setting I. The two columns under the Carved and Split estimators report the mean and median bias and squared errors respectively. 
The cells in the single column under Debiased report the median bias and squared errors, due to reasons indicated in the description.}}
\label{tab:1}
\end{table}

\begin{table}[h!]
\def~{\hphantom{0}}
\centering
\resizebox{\columnwidth}{!}{\begin{tabular}{cccccccccccccc}
        & &  \multicolumn{5}{c}{Bias} &   \multicolumn{5}{c}{MSE}   \\
    $\sigma_{\varepsilon} $   & $K$ & \multicolumn{2}{c}{Carved} & \multicolumn{2}{c}{Split} & Debiased  &  \multicolumn{2}{c}{Carved} & \multicolumn{2}{c}{Split} & Debiased \\
4 & 3 & $0.066$ & $0.065$ & $0.088$ & $0.089$ & $-0.081$ &   $0.024$ & $0.011$ & $0.056$ & $0.026$ & $0.059$  \\
4 & 5 &  $0.066$ & $0.072$ & $0.069$ & $0.063$ & $-0.083$ &   $0.016$ & $0.007$ & $0.033$ & $0.013$ & $0.045$   \\
4 & 10 & $0.070$ & $0.072$ & $0.079$ & $0.079$ & $-0.079$ &   $0.011$ & $0.006$ & $0.022$ & $0.009$ & $0.016$  \\
2 & 3 & $0.016$ & $0.012$ & $0.026$ & $0.023$ & $-0.089$ &   $0.005$ & $0.002$ & $0.012$ & $0.005$ & $0.020$  \\
2 & 5 & $0.016$ & $0.018$ & $0.019$ & $0.017$ & $-0.075$ &   $0.005$ & $0.001$ & $0.007$ & $0.003$ & $0.013$  \\
2 & 10 &  $0.016$ & $0.015$ &  $0.020$ & $0.021$ & $-0.061$ &   $0.002$ & $0.001$ & $0.004$ & $0.002$ & $0.006$  \\
    \end{tabular}}
    \vspace{0.2cm}
\caption{\normalfont{Simulation results under Setting II. The two columns under the Carved and Split estimators report the mean and median bias and squared errors respectively. The cells in the single column under Debiased report the median bias and squared errors.}}
\label{tab:2}
\end{table}

\begin{table}[h!]
\def~{\hphantom{0}}
\centering
\resizebox{\columnwidth}{!}{\begin{tabular}{cccccccccccccc}
        & &  \multicolumn{5}{c}{Bias} &   \multicolumn{5}{c}{MSE}   \\
    $\sigma_{\varepsilon} $   & $K$ & \multicolumn{2}{c}{Carved} & \multicolumn{2}{c}{Split} & Debiased  &  \multicolumn{2}{c}{Carved} & \multicolumn{2}{c}{Split} & Debiased \\
4 & 3 & $0.102$ & $0.103$ & $-0.179 $ & $-0.164$ & $-0.096$ &   $0.053$ & $0.020$ & $0.084$ & $0.036$ & $0.067$ \\
4 & 5 & $0.110$ & $0.112$ & $-0.170$ & $-0.167$ & $-0.091$ &   $0.042$ & $0.015$ & $0.067$ & $0.033$ & $0.042$   \\
4 & 10 & $0.092$ & $0.100$ & $-0.163$ & $-0.162$ & $-0.092$ &   $0.023$ & $0.011$ & $0.042$ & $0.027$ & $0.026$  \\
2 & 3 & $0.028$ & $0.024$ & $-0.156$ & $-0.151$ &  $-0.079$ &   $0.008$ & $0.003$ & $0.042$ & $0.024$ & $0.018$   \\
2 & 5 & $0.024$ & $0.034$ & $-0.156$ & $-0.159$ & $-0.075$ &   $0.016$ & $0.002$ & $0.034$ & $0.025$ & $0.013$  \\
2 & 10 & $0.029$ & $0.029$ & $-0.152$ & $-0.151$ & $-0.067$ &   $0.003$ & $0.001$ & $0.038$ & $0.023$ & $0.008$  \\
    \end{tabular}}
    \vspace{0.2cm}
\caption{\normalfont{Simulation results under Setting III. The two columns under the Carved and Split estimators report the mean and median bias and squared errors respectively. The cells in the single column under Debiased report the median bias and squared errors.}}
\label{tab:3}
\end{table}

The cells in Tables \ref{tab:1}, \ref{tab:2}, and \ref{tab:3} summarize our findings in the three different settings.
In all the tables, our proposed estimator is indicated as ``Carved", while the split-and-aggregate estimator is called ``Split'' and the aggregated debiased LASSO estimator is called ``Debiased".
Based on  Table \ref{tab:1}, we observe that when the noise level is low and the treatment variable is not highly correlated with the covariates, the selected model includes all the confounding variables with a high probability, and therefore, both the carved estimator and the split-and-aggregate estimator have little bias in estimating $\alpha$. 
When the noise level is higher, the two estimators tend to have slightly more bias than the debiased LASSO estimator. 
But, the bias of the estimators is still dominated by their variance in the mean squared errors.
On the other hand, the carved estimator has smaller mean-squared error than the split estimator in all the considered cases, which is expected since the carved estimator  integrates information from the existing studies with higher efficiency. 
We note that the aggregated debiased LASSO estimator loses much efficiency relative to the aggregated carved estimator when the treatment variable is highly correlated with some of the covariates, which is shown in Tables \ref{tab:2} and \ref{tab:3}.
The empirical comparisons here are in line with what we expect from our theoretical investigation.  

\section{Synthetic study on clinical trial data}
\label{sec:real}

We demonstrate the efficacy of our aggregation and estimation scheme on synthetic data generated from the Afya II clinical trial. 
Note that our paper proposes a new design to collect data for conducting validation studies and pooling data from already existing studies to estimate treatment effects.
Therefore, collecting data that aligns with our recommendation for new study designs and verifying the effectiveness of our proposal is 
beyond the scope of this paper.
Instead, we evaluate the efficiency gains of our estimator through a ``realistic" simulation study that very closely resembles the Afya II clinical trial data.
In our synthetic case study, we assume the following scenario. 
Two previous clinical trial studies have already been conducted to assess the effectiveness of an intervention, and we seek to design a new trial to confirm the results of the earlier clinical trials.

The main goal of the Ayfa II clinical trial was to test the effectiveness of a cash incentive on enhancing HIV patient adherence to antiretroviral (ARV) therapy and on achieving HIV viral suppression. The actual trial enrolled a total of $534$ HIV-positive patients, who were randomly assigned to one of three arms: a control arm, a treatment arm with a cash incentive amount of approximately $5$ US dollars, and a treatment arm with a cash incentive amount of approximately $11$ US dollars. 
The primary outcome was determined by measuring the HIV viral load in patient plasma specimens 12 months after the treatment assignment. During this study, the trial collected patient status information, including age, sex, marital status (single, married/cohabitating, divorced/separated, widowed, or unknown), weight (measured in kilograms), WHO clinical stage (ranging from 1 to 4), pregnancy status, nutrition status (normal, obesity, or malnourished), family planning ID (with a total of 5 possible statuses), ARV status (start ARV, continue, change, or stop), TB screening, TB status, and initial CD4 count.
Sociodemographic survey data was collected at the end of the study. Overall, the trial data measured $92$ covariates for the $534$ enrolled patients.

In our case study, we focus on the treatment arm where patients received a cash incentive of $11$ US dollars.
This reduces the sample size to $360$ patients. 
First, we use the LASSO to regress the outcome (viral load) against the treatment indicator variable and the collected covariates on the trial data.
The LASSO selects a total of $12$ covariates.
We record the estimated sparse coefficients as $\hat{\beta}{\texttt{Lasso}}$.
Using covariate measurements from the actual trial,
we generate the viral load outcome from a linear model 
$$Y = 10\cdot D + X \hat{\beta}_{\texttt{Lasso}} + \varepsilon,$$
where $\varepsilon \in \mathbb{R}^{360}$ is a multivariate Gaussian random variable with mean zero and identity covariance matrix.
Finally, we randomly divide the synthetic dataset into three folds, treating two of them as data collected from existing studies and the remaining one as data from our validation study. 
We repeat this entire process $500$ times by generating data from the above linear model.
This gives us $500$ synthetic datasets that are closely modeled along the actual Afya II clinical trial.
During estimation, we use only the selected covariates from any of the two existing studies as part of our validation data.

In line with the preceding section, Table \ref{tab:4} reports the comparison between our estimator, and the split-and-aggregate estimator and the aggregated debiased LASSO estimator.
As done before, we compare their bias and mean squared errors. 
We see that the carved estimator, based on aggregating only summary statistics from the two existing studies, achieves the lowest bias and smallest mean squared error. 
More specifically, the improvement in mean squared error over the debiased LASSO estimator is roughly around $15\%$, and is a significant $60\%$ over the split-and-aggregate estimator.
This case study on clinical trial data re-emphasizes that our aggregation scheme can be used to design validation studies from a synthesis of existing studies, and that we can estimate the common treatment effects without sacrificing statistical efficiency.

\begin{table}[h]
    \centering
    \fontsize{10pt}{10pt}\selectfont
    \begin{tabular}{cccccc}
          \multicolumn{3}{c}{Bias} &  \multicolumn{3}{c}{MSE}  \\
          Carved & Split & Debiased  & Carved & Split & Debiased \\ 
          0.001 & -0.011 & -0.062 & 0.058 & 0.154 & 0.068
    \end{tabular}
     \caption{\centering{Results on $500$ synthetic datasets modeled along the Afya II clinical trial.}}
    \label{tab:4}
\end{table}




\section{Concluding Remarks}
\label{Sec:conclusion}

In this paper, we develop a new scheme for estimating common treatment effects from a synthesis of prior studies.
Our scheme applies ideas from data carving to: (i) use already existing data for an unbiased estimation of treatment effects in selected models, and (ii) aggregate summary statistics from existing studies when the LASSO is used in each individual study to select potential confounding variables.
As a result, we have laid out a data aggregation and validation analysis protocol that gives us efficient unbiased estimators while preserving the privacy of individual records. 
The summary statistics required by our estimator include the first two moments in each selected model along with some compressed information from the model selection.
Our estimator, unlike the debiased LASSO estimator, does not require us to store, report, or estimate the inverse of a high-dimensional matrix. 
As shown by our simulation studies, the introduced estimator can not just mitigate bias from model selection, but also yield higher statistical efficiency than popular alternatives including off-the-shelf alternatives such as splitting.
Interval estimation, which requires us to characterize the distribution of the aggregated estimator, is a challenging question and a natural next goal.
We hope to address the problem of interval estimation with our new estimator in future work.


\begin{funding}
The first author is supported by NSF Grants DMS-1951980 and DMS-2113342. The second author is supported in part by NSF Grant DMS-2015325 and NIH Grant R01MH125746. The third author is supported in part by NSF Grants DMS-1914496 and DMS-1951980. 
\end{funding}

\bibliographystyle{imsart-number} 

\bibliography{reference}

\appendix

\section{Proofs}
\label{Sec:proofs}

\subsection{Carved Estimator: first example}

Suppose that $\omega_1$ and $\widehat{\gamma}_1$ are Gaussian variables with the distribution stated in \ref{prop:1} and \ref{prop:2}, i.e., the limiting distribution stated for these variables in \ref{prop:1} and \ref{prop:2} holds exactly.
We note the following.
\begin{proposition}
	\label{exact:UMVU:simple}
	Let $\widehat\alpha_1^{\text{\;carve}}$ assume the value
	\begin{equation*}
	\begin{aligned}
	& \widehat{\alpha}_1 + \Bigg\{ 1_{\mathbb{R}^+}(s_{E_1})\cdot  (1-\rho^2)^{-1/2}\dfrac{\rho}{\sqrt{N_1}} \left(\bar{\Phi}\left(\dfrac{(\lambda-(1-\rho^2) \sqrt{N_1}\widehat{\beta}_1)}{\sqrt{(1-\rho^2)}}\right)\right)^{-1}{\phi\left(\dfrac{(\lambda-(1-\rho^2) \sqrt{N_1}\widehat{\beta}_1)}{\sqrt{(1-\rho^2)}}\right)} \\
	& - 1_{\mathbb{R}^-}(s_{E_1}) \cdot  (1-\rho^2)^{-1/2}\dfrac{\rho}{\sqrt{N_1}} \left(\Phi\left(\dfrac{(-\lambda-(1-\rho^2) \sqrt{N_1}\widehat{\beta}_1)}{\sqrt{(1-\rho^2)}}\right)\right)^{-1}{\phi\left(\dfrac{(\lambda+(1-\rho^2) \sqrt{N_1}\widehat{\beta}_1)}{\sqrt{(1-\rho^2)}}\right)}\Bigg\}. 
	\end{aligned}
	\end{equation*}	  
	Then, the following holds
	$$\mathbb{E}\left[\widehat{\alpha}_1^{\text{\; carve}}\right] = \alpha.$$
\end{proposition}

\begin{proof}
	Our proof proceeds in two steps.
	
	\smallskip
	\emph{Step 1}: \ \ Fix
	$$\hat\alpha_1^{\text{\;initial}}=\hat{\alpha}_1 - (1-\rho^2)^{-1} \dfrac{1}{\sqrt{N_1}}\cdot (\omega_{1,1}- \rho \omega_{1,2}).$$
	Note that
	$$
	\mathbb{E}\left[ \hat\alpha_1^{\text{\; initial}}\; \Big\lvert \; \widehat{E}_1 =  E_1, \;\text{sign}(\hatB_1) = s_{E_1}\right] = \alpha.
	$$
	This observation follows by writing our conditioning event as
	\begin{equation}
	\label{cond:event:simple}
	\left\{ \widehat{E}_1 = E_1,  \text{sign}(\hatB_1) =s_{E_1} \right\}  =  \left\{s_{E_1} (\omega_{1,2}-\rho\omega_{1,1})+s_{E_1} (1-\rho^2) \sqrt{N_1}\widehat{\beta}_1 -\lambda >0\right\}. 
	\end{equation}
	Observe that the initial estimate is independent of
	$$s_{E_1} (\omega_{1,2}-\rho\omega_{1,1})+s_{E_1} (1-\rho^2) \sqrt{N_1}\widehat{\beta}_1 -\lambda.$$
	Thus, we note that $\hat\alpha_1^{\text{\; initial}}$ is independent of our conditioning event, which 
	gives us 
 $$\mathbb{E}\Big[ \hat{\alpha}_1^{\text{initial}} \ \Big\lvert \widehat{E}_1 = E_1,  \text{sign}(\hatB_1) =s_{E_1}\Big]= \mathbb{E}\Big[ \hat{\alpha}_1^{\text{initial}}\Big].$$
 Because 
 $$\mathbb{E}\Big[ \hat{\alpha}_1^{\text{initial}}\Big]= \alpha$$ for fixed set $E_1=\{1\}$, we have an initial unbiased estimator.
	
	\emph{Step 2}: \ \ Conditioning further upon $\widehat{\gamma}_1$ provides us the estimator
	\begin{equation*}
	\begin{aligned}
	&\hat\alpha_1^{\text{\;carve}}
	= \mathbb{E}\Big[\widehat{\alpha}_1 -(1-\rho^2)^{-1}\dfrac{1}{\sqrt{N_1}}\cdot(\omega_{1,1}- \rho \omega_{1,2})\ \Big\lvert \widehat{\gamma}_1, s_{E_1} (\omega_{1,2}-\rho\omega_{1,1})\\
	&\;\;\;\;\;\;\;\;\;\;\;\;\;\;\;\;\;\;\;\;\;\;\;\;\;\;\;\;\;\;\;\;\;\;\;\;\;\;\;\;\;\;\;\;\;\;\;\;\;\;\;\;\;\;\;\;\;\;\;\;\;\;\;\;\;\;\;\;\;\;\;\;\;\;\;\;\;\;\; \;\;\;\;\;\; +s_{E_1} (1-\rho^2) \sqrt{N_1}\widehat{\beta}_1 -\lambda>0\Big]\\
	&= \widehat{\alpha}_1 - (1-\rho^2)^{-1} \dfrac{1}{\sqrt{N_1}}\cdot\mathbb{E}\Big[\omega_{1,1}- \rho \omega_{1,2} \  \lvert \  \widehat{\gamma}_1,\; s_{E_1}(\omega_{1,2}-\rho\omega_{1,1})\\
	&\;\;\;\;\;\;\;\;\;\;\;\;\;\;\;\;\;\;\;\;\;\;\;\;\;\;\;\;\;\;\;\;\;\;\;\;\;\;\;\;\;\;\;\;\;\;\;\;\;\;\;\;\;\;\;\;\;\;\;\;\;\;\;\;\;\;\;\;\;\;\;\;\;\;\;\;+s_{E_1}(1-\rho^2) \sqrt{N_1}\widehat{\beta}_1 -\lambda>0\Big],
	\end{aligned}
	\end{equation*}
	an improvement over $\hat\alpha_1^{\text{\; initial}}$ through Rao-Blackwellization.
	
	We compute the conditional expectation:
	$$\mathbb{E}\Big[\omega_{1,1}- \rho \omega_{1,2} \  \lvert \  \widehat{\gamma}_1,\; s_{E_1}(\omega_{1,2}-\rho\omega_{1,1})+s_{E_1}(1-\rho^2) \sqrt{N_1}\widehat{\beta}_1 -\lambda>0\Big],$$
	to obtain a final expression for our estimator.
	Note that we have, by construct, 
	$$\mathbb{E}\left[\widehat{\alpha}_1^{\text{\; carve}}\;\Big\lvert \; \widehat{E}_1 =  E_1, \;\text{sign}(\hatB_1) = s_{E_1}\right]=\alpha.$$
	Our claim in the Proposition now follows from the tower property of expectation.
\end{proof}

\subsection{Carved Estimator: general problem}

We give the proofs of the results in Section \ref{Sec:theory}.
\begin{proof}[Proof of Lemma \ref{limiting:Gaussianity}]
First, let
$$\widehat{\Sigma}_{-k,k} = \frac{1}{N_k}\left\{X_{E_k^c}' \begin{pmatrix} D & X_{E_k}  \end{pmatrix} + X_{k,E_k^c}' \begin{pmatrix} D_k & X_{k, E_k}  \end{pmatrix}\right\}$$
and let $\Sigma_{-k,k}=\mathbb{E}[\widehat\Sigma_{-k,k}]$.
Also, let $\Sigma_k =\mathbb{E}[\widehat{\Sigma}_k]$ where
$$\hat{\Sigma}_{k} = \frac{1}{N_k}\left(n_k \hat{\Xi}_{k}  + \begin{pmatrix} D & X_{E_k} \end{pmatrix}' \begin{pmatrix} D & X_{E_k} \end{pmatrix}\right).$$
Then, define
	\begin{equation*}
	\begin{aligned}
\bar{T}_{N_k}	&= \frac{1}{N_k} \begin{pmatrix} \Sigma^{-1}_k \begin{pmatrix} D & X_{E_k} \end{pmatrix}'(y-D \alpha - X_{E_k}\beta_{E_k}) \\ (X'_{E_k^c}-\Sigma_{-k,k}\Sigma^{-1}_{k}\begin{pmatrix} D & X_{E_k} \end{pmatrix}')(y-D \alpha - X_{E_k}\beta_{E_k})\\ -\begin{pmatrix} D & X \end{pmatrix}'(y-D \alpha - X_{E_k}\beta_{E_k}) \end{pmatrix}\\
	&\;\;\;\;\;+ \frac{1}{N_k}\begin{pmatrix} \Sigma^{-1}_k \begin{pmatrix} D_k & X_{k,E_k} \end{pmatrix}'(y_k-D_k \alpha - X_{k,E_k}\beta_{E_k}) \\ (X'_{k,E_k^c}-\Sigma_{-k,k}\Sigma^{-1}_{k}\begin{pmatrix} D_k & X_{k,E_k} \end{pmatrix}')(y_k-D_k \alpha - X_{k,E_k}\beta_{E_k})\\ \frac{1}{r_k} \begin{pmatrix} D_k & X_{k} \end{pmatrix}'(y_k-D_k \alpha - X_{E_k}\beta_{E_k}) -\begin{pmatrix} D_k & X_{k} \end{pmatrix}' (y_k-D_k \alpha - X_{k,E_k}\beta_{E_k})\end{pmatrix}.
	\end{aligned}
	\end{equation*}
The proof of the Lemma closely follows Proposition 4.1 in \cite{panigrahi2016integrative}, and hence we omit further details here.
\end{proof}

We state a useful Lemma before providing a proof for Proposition \ref{exact:UMVU}.
Recall that
$$\widehat{\Sigma}_{-k,k} = \frac{1}{N_k}\left\{X_{E_k^c}' \begin{pmatrix} D & X_{E_k}  \end{pmatrix} + X_{k,E_k^c}' \begin{pmatrix} D_k & X_{k, E_k}  \end{pmatrix}\right\},$$
and $\Sigma_{-k,k}=\mathbb{E}[\widehat\Sigma_{-k,k}]$.
Suppose that we rewrite \eqref{eq:lasso} as
\begin{equation}
\label{lasso:randomized:rep:gen}
\underset{\alpha, \beta}{\text{minimize}} \;\; \dfrac{1}{2\sqrt{N_k}}\|Y -D \alpha - X\beta \|^2  + \dfrac{1}{2\sqrt{N_k}}\|Y_k - D_k \alpha-  X_k\beta \|^2- \w_{k,1}^T \alpha -\w_{k,2}^T \beta + \|\Lambda \beta\|_1
\end{equation}
by using the randomization variable in \eqref{randomization:gen}.
Fixing some more symbols, we let
$$\zeta_k = \begin{pmatrix} 0'_{s} & (\Lambda_{E_k} s_{E_k})' \end{pmatrix}',$$
and denote by $Z_k$ the inactive components of subgradient variables from the $\ell_1$ penalty at  the LASSO solution.
Note that the K.K.T. mapping for \eqref{lasso:randomized:rep:gen} is equal to
\begin{equation}
\label{stationary:map}
\w_{k}=\begin{bmatrix}-\widehat{\Sigma}_{k} & 0\\ -\widehat{\Sigma}_{-k,k} & -I \end{bmatrix}\sqrt{N_k}\begin{pmatrix} \hat{\gamma}_k \\  \hat{\Gamma}_{k} \end{pmatrix} + \begin{bmatrix} \widehat{\Sigma}_{k}& 0 \\ \widehat{\Sigma}_{-k,k} & \Lambda_{E_k^c}\end{bmatrix}  \begin{pmatrix} \sqrt{N_k} \hatG_k \\   Z_k \end{pmatrix} + \begin{pmatrix} \zeta_k \\ 0 \end{pmatrix}, 
\end{equation}
where the variables $\hatG_k$ and $Z_k$ satisfy
$$\|Z_k \|\leq 1, \ \text{sign}(\hatB_k) = s_{E_k}.$$

The next result derives an expression for the marginal asymptotic likelihood based on the limiting distribution of $\widehat{\gamma}_k$ and $\hatG_k$, before we condition on the event of selection.

\begin{lemma}
Marginally, the variables $\widehat{\gamma}_k$, $\hatG_k$ are asymptotically independent of 
$\widehat{\Gamma}_k$ and $Z_k$. 
The marginal asymptotic likelihood, derived from the limiting distribution of $\widehat{\gamma}_k$ and $\hatG_k$, is proportional to
\begin{equation*}
	\begin{aligned}
 & \exp\Bigg(-\frac{N_k}{2}(\widehat{\gamma}_k- \gamma_k)' \Sigma_k(\widehat{\gamma}_k- \gamma_k)\Bigg)\cdot \exp\Bigg(-(1-r_k)^{-1}r_k\frac{N_k}{2}\left(\hatG_k- \hat\gamma_k + \frac{1}{\sqrt{N_k}}\Sigma^{-1}_{k} \zeta_k\right)' \\
 &\;\;\;\;\;\;\;\;\;\;\;\;\;\;\;\;\;\;\;\;\;\;\;\;\;\;\;\;\;\;\;\;\;\;\;\;\;\;\;\;\;\;\;\;\;\;\;\;\;\;\;\;\;\;\;\;\;\;\;\;\;\;\;\;\;\;\;\;\;\;\;\;\;\;\;\;\;\;\;\;\;\;\;\;\;\;\;\;\;\;\;\;\;\;\;\;\;\;\;\;\;\;\;\;\; \Sigma_k\left(\widehat{\gamma}_k- \hat\gamma_k + \frac{1}{\sqrt{N_k}}\Sigma^{-1}_{k} \zeta_k\right)\Bigg).
   \end{aligned}
\end{equation*}
\label{lemma:limit:properties}
\end{lemma}

\begin{proof}
Using \eqref{stationary:map} together with the fact that $\widehat\Sigma_{k}- \Sigma_k= o_p(1)$, $\widehat\Sigma_{-k,k}- \Sigma_{-k,k}=o_p(1)$ in our i.i.d. framework, we deduce that the variables $\hatG_k$ and $Z_k$ in this map admit the asymptotic representation:
\begin{equation}
\label{aymp:rep}
\begin{aligned}
\sqrt{N}_k\hatG_k &=\Sigma^{-1}_{k}\w_{k, E} + \sqrt{N_k}\hat{\gamma}_k -\Sigma^{-1}_{k}\zeta_k + o_p(1);\\
Z_k  &=(\Lambda_{E_k^c})^{-1}(\w_{k, E^c} -\Sigma_{-k,k}\Sigma^{-1}_{k}\w_{k, E}  + \Sigma_{-k,k}\Sigma^{-1}_{k}\zeta_k  +\sqrt{N}_k \hat\Gamma_k) + o_p(1).
\end{aligned}
\end{equation}
We easily verify our claim by using the limiting marginal distribution of $\sqrt{N_k}(\widehat{\gamma}_k- \gamma_k)$, $\sqrt{N_k}\widehat{\Gamma}_k$, and $\w_k$, that was  derived in Lemma \ref{limiting:Gaussianity}.
\end{proof}

We are ready to provide the proof of Proposition \ref{exact:UMVU}.
The limiting distribution in Lemma \ref{limiting:Gaussianity} yields us an asymptotic debiasing term for the least squares estimator in the selected model after we condition on the selection event. 

\begin{proof}\emph{Proposition \ref{exact:UMVU}}.
	We consider the following initial estimator 
	$$
	\hat\alpha_k^{\text{\; initial}} = (1-r_k)^{-1} \hat\alpha_k - (1-r_k)^{-1} r_k e_1'\left( \frac{1}{\sqrt{N_k}}\Sigma^{-1}_{k} \zeta_k +\hatG_{k} \right).
	$$
	Observe, we have
	$$\mathbb{E}\left[ \hat\alpha_k^{\text{\; initial}} \; \Big\lvert \; \widehat{E}_k =  E_k, \;\text{sign}(\hatB_k) = s_{E_k}\right] = \alpha.$$
	This is because:
	\begin{equation*}
	\begin{aligned}
	& \mathbb{E}\left[ \hat\alpha_k^{\text{\; initial}} \; \Big\lvert \; \widehat{E}_k =  E_k, \;\text{sign}(\hatB_k) = s_{E_k}\right]\\
	&=  \mathbb{E}\left[  (1-r_k)^{-1} \hat\alpha_k - (1-r_k)^{-1} r_k e_1'\left\{ \frac{1}{\sqrt{N_k}}\Sigma^{-1}_{k} \zeta_k  +\hatG_{k}\right\}\right]\\
	&= (1-r_k)^{-1}  \mathbb{E}\left[ \hat\alpha_k\right] - (1-r_k)^{-1} r_k e_1'\left\{\frac{1}{\sqrt{N_k}}\Sigma^{-1}_{k} \zeta_k  + \mathbb{E}\left[ \mathbb{E}\left[\hatG_{k}\;\lvert \; \hat\gamma_k\right]\right]\right\}\\
	&= (1-r_k)^{-1}  \mathbb{E}\left[ \hat\alpha_k\right] - (1-r_k)^{-1} r_k e_1' \left\{\frac{1}{\sqrt{N_k}}\Sigma^{-1}_{k} \zeta_k  + \mathbb{E}\left[  \hat\gamma_k - \frac{1}{\sqrt{N_k}}\Sigma^{-1}_{k} \zeta_k \right]\right\}\\
	&=(1-r_k)^{-1} \mathbb{E}\left[ \hat\alpha_k\right]  - (1-r_k)^{-1} r_k e_1'\mathbb{E}\left[\hat\gamma_k\right] = \alpha.
	\end{aligned}
	\end{equation*}
	Note, the first display uses the fact that our selection event can be rewritten as 
 $$\left\{\|Z_k \|_{\infty}\leq 1, \ \text{sign}(\hatB_k)= s_{E_k}\right\} .$$
  That is,
	the conditional expectation of the initial estimator is equal to
	$$
	\mathbb{E}\left[ (1-r_k)^{-1} \hat\alpha_k - (1-r_k)^{-1} r_k e_1'\left\{ \frac{1}{\sqrt{N_k}}\Sigma^{-1}_{k} \zeta_k  +\hatG_{k} \right\}
	\; \Big\lvert \;  \|Z_k \|_{\infty}\leq 1, \ \text{sign}(\hatB_k)= s_{E_k}\right].
	$$
	The above expression is further equal to
	$$
	\mathbb{E}\left[ (1-r_k)^{-1} \hat\alpha_k - (1-r_k)^{-1} r_k e_1'\left\{ \frac{1}{\sqrt{N_k}}\Sigma^{-1}_{k} \zeta_k  +\hatG_{k} \right\}\right],
	$$
	due to the independence between $Z_k$ and $\hat\alpha_k^{\text{\; initial}}$ and between $\hatG_{k}$ and $\hat\alpha_k^{\text{\; initial}}$ in the asymptotic limit, which follows directly from Lemma \ref{lemma:limit:properties}.
	The second display uses the tower property of expectation.
	Using Lemma \ref{lemma:limit:properties} once again, we note that
	$$\mathbb{E}\left[\hatG_{k}\;\Big\lvert \; \hat\gamma_k\right]=  \hat\gamma_k - \frac{1}{\sqrt{N_k}}\Sigma^{-1}_{k} \zeta_k,$$
	in the penultimate display.
	Conditioning further upon $\hat\gamma_k$, the complete sufficient statistic, our estimator in the claim is equal to
	\begin{equation}
	\label{UB:final}
	\widehat{\alpha}_k^{\text{\; carve}} = \mathbb{E}\left[  \hat\alpha_k^{\text{\; initial}}\; \Big\lvert \; \hat\gamma_k,\; \widehat{E}_k =  E_k, \;\text{sign}(\hatB_k) = s_{E_k}\right].
	\end{equation}
	Obtaining an expression for the right-hand side estimate in \eqref{UB:final}, we have
	\begin{equation*}
	\begin{aligned}
	& \mathbb{E}\left[ \hat\alpha_k^{\text{\; initial}} \; \Big\lvert \; \hat\gamma_k,\; \widehat{E}_k =  E_k, \;\text{sign}(\hatB_k) = s_{E_k} \right]\\
	&=  (1-r_k)^{-1}\cdot \hat\alpha_k   - (1-r_k)^{-1} r_k e_1'\frac{1}{\sqrt{N_k}}\Sigma^{-1}_{k} \zeta_k\\
	&\;\;\;\;\;\;\;\;\;\;\;\;\;\;\;\;\;\;\;\;\;\;\;\;\;\;\;\;\;\;\;\;\;\;\;\;\;\;\;\;- (1-r_k)^{-1} r_k\frac{1}{\sqrt{N_k}}\cdot  e_1'\mathbb{E}\left[  \sqrt{N_k}\hatG_{k}\; \Big\lvert \; \hat\gamma_k,\;\text{sign}(\hatB_k) = s_{E_k}  \right].
	\end{aligned}
	\end{equation*}
	Because 
	$$\mathbb{E}\left[  \sqrt{N_k}\hatG_{k}\; \Big\lvert \; \hat\gamma_k,\;\text{sign}(\hatB_k) = s_{E_k}  \right]=\mu_{\mathcal{H}}\left(\sqrt{N_k}\hat\gamma_k -\Sigma^{-1}_{k} \zeta_k, r_k^{-1}(1-r_k)\Sigma^{-1}_k\right),$$
	we obtain our estimator in the claim.
\end{proof}

\subsection{Asymptotic properties of carved estimator}
 
Before we provide a proof for the asymptotic unbiasedness of our estimator, we present a limiting value for the probability of selection.
Fixing some notations, set  $\gamma_k = \begin{pmatrix} \alpha' & \beta'_{E_k} \end{pmatrix}'$.
As per our parameterization in \eqref{seq:parameters}
$$\sqrt{N_k}\gamma_k = a_{N_k}\gamma_0.$$ 

\begin{proposition}
	\label{conv:log:partition}
	Suppose that the assumptions in Conditions \ref{moment:condition:0} and \ref{prob:sel:condition:0} are met. 
	Then, the probability of selection assumes the following limiting value:
	\begin{equation*}
	\begin{aligned}
	&\lim_{N_k\to \infty} \frac{1}{(a_{N_k})^{2}}\log \mathbb{P}\left[\widehat{E}_k =  E_k, \;\text{sign}(\hatB_k) = s_{E_k}\right]+\displaystyle\inf_{\widetilde{\gamma}, \widetilde{z}}\Big\{ \dfrac{1}{2}(\widetilde{\gamma}- \gamma_0)' \Sigma_{k} (\widetilde{\gamma}- \gamma_0) \\
	&+  \dfrac{1}{2}(1-r_k)^{-1}r_k (\widetilde{z}-\widetilde{\gamma}+  \frac{1}{a_{N_k}}\Sigma_{k}^{-1}\zeta_k)'\Sigma_{k}  (\widetilde{z}-\widetilde{\gamma} + \frac{1}{a_{N_k}} \Sigma_{k}^{-1}\zeta_k)+  \dfrac{1}{(a_{N_k})^2}B_{s_{E_k}}(a_{N_k}\widetilde{z})\Big\} -C_0=0,
	\end{aligned}
	\end{equation*}
 where $C_0$ is a constant free of $\gamma_0$.
\end{proposition}

\begin{proof}
	We have 
	$$\lim_{N_k\to \infty}\frac{1}{(a_{N_k})^{2}}\left\{\log \mathbb{P}[\widehat{E}_k =  E_k, \;\text{sign}(\hatB_k) = s_{E_k}]-\log \mathbb{P}[ \|Z_k \|_{\infty}\leq 1, \ \text{sign}(\hatB_k)= s_{E_k}]\right\}=0,$$
	based on the K.K.T. conditions for the LASSO in  \eqref{stationary:map}.
	Using the asymptotic representation for the variables $Z_k$ and $\hatG_k$ in \eqref{aymp:rep} and the condition in Condition \ref{prob:sel:condition:0}, we note that 
	$$\lim_{N_k\to \infty}\frac{1}{(a_{N_k})^{2}} \log \mathbb{P}[ \|Z_k \|_{\infty}\leq 1, \ \text{sign}(\hatB_k)= s_{E_k}]$$
	is equal to
	\begin{equation*}
	\begin{aligned}
	& \lim_{N_k\to \infty}\frac{1}{(a_{N_k})^{2}}\Big\{\log \mathbb{P}\left[ \text{sign}(e_2'(\Sigma^{-1}_{k}\w_{k, E} + \sqrt{N_k}\hat{\gamma}_k -\Sigma^{-1}_{k}\zeta_k)) = s_{E_k}\right]\\
	&\;\;\;\;\;\;\;\;\;\;\;\;\;\;\;\;\;\;\;+ \log \mathbb{P}\left[ \|(\Lambda_{E_k^c})^{-1}(\w_{k, E^c} -\Sigma_{-k,k}\Sigma^{-1}_{k}\w_{k, E}  + \Sigma_{-k,k}\Sigma^{-1}_{k}\zeta_k  +\sqrt{N}_k \hat\Gamma_k)\|_{\infty}\leq 1\right]\Big\}\\
	&= \lim_{N_k\to \infty} T_{1,n} + T_{2,n}
	\end{aligned}
	\end{equation*}
	where, for $v_1 \in \mathbb{R}^s$ and $v_2 \in \mathbb{R}^{q_k}$, $e_2' \begin{pmatrix} v_1 & v_2 \end{pmatrix} = v_2$. 
    The assumption in Condition \ref{moment:condition:0} allows to apply a  large deviation limit for the two log-probabilities $T_{1,n}$, $T_{2,n}$.
    First, we note that  
    $$\lim_{N_k\to \infty} T_{1,n}= \lim_{N_k\to \infty}\frac{1}{(a_{N_k})^{2}}\log \mathbb{P}\left[ \text{sign}(e_2'(\Sigma^{-1}_{k}\w_{k, E} + \sqrt{N_k}\hat{\gamma}_k -\Sigma^{-1}_{k}\zeta_k)) = s_{E_k}\right]$$
    agrees with
	\begin{equation}
	  -\displaystyle\inf_{\widetilde{\gamma}, w_E:  \text{sign}(e_2'(\Sigma^{-1}_{k}w_E + \widetilde{\gamma})) = s_{E_k}} \ \Big\{\dfrac{1}{2}(\widetilde{\gamma}- \gamma_0)' \Sigma_{k} (\widetilde{\gamma}- \gamma_0)  +  (1-r_k)^{-1}r_k\dfrac{1}{2}w_E' \Sigma^{-1}_{k} w_E\Big\}  \label{imp:equa}
	\end{equation}
	up to an additive constant.
	Second, we note that the limiting distribution of the variables in $T_{2,n}$ is free from $\gamma_0$, and hereafter, let $\displaystyle\lim_{N_k\to \infty} T_{2,n}=\widetilde{C}$.

	At last, we apply a reparameterization to the optimizing variables
	$w_E\to \widetilde{z}$ such that
	$\widetilde{z} =\Sigma^{-1}_{k}w_{E} + \widetilde{\gamma} $ in \eqref{imp:equa}.
	Thus, we note that the limiting value of the log-probability of selection is equal to 
	\begin{equation*}
	\begin{aligned}
	& -\displaystyle\inf_{\widetilde{\gamma}, \widetilde{z}: \text{sign}(\widetilde{z}_2)=s_{E_k}} \ \Big\{\dfrac{1}{2}(\widetilde{\gamma}- \gamma_0)' \Sigma_{k} (\widetilde{\gamma}- \gamma_0) + \dfrac{1}{2}(1-r_k)^{-1}r_k (\widetilde{z}-\widetilde{\gamma})'\Sigma_{k}  (\widetilde{z}-\widetilde{\gamma})\Big\} + C_0,
	\end{aligned}
	\end{equation*}
 where $C_0$ is a constant free of $\gamma_0$.
	Because of convexity of the above optimization objective, we conclude that
	\begin{equation*}
	\begin{aligned}
	&\lim_{N_k\to \infty} \Big\{\frac{1}{(a_{N_k})^{2}}\log \mathbb{P}\left[\widehat{E}_k =  E_k, \;\text{sign}(\hatB_k) = s_{E_k}\right]+\displaystyle\inf_{\widetilde{\gamma}, \widetilde{z}}\Big\{ \dfrac{1}{2}(\widetilde{\gamma}- \gamma_0)' \Sigma_{k} (\widetilde{\gamma}- \gamma_0) \\
	&+  \dfrac{1}{2}(1-r_k)^{-1}r_k (\widetilde{z}-\widetilde{\gamma}+  \frac{1}{a_{N_k}}\Sigma_{k}^{-1}\zeta_k)'\Sigma_{k}  (\widetilde{z}-\widetilde{\gamma} + \frac{1}{a_{N_k}} \Sigma_{k}^{-1}\zeta_k)+  \dfrac{1}{(a_{N_k})^2}B_{s_{E_k}}(a_{N_k}\widetilde{z})\Big\} -C_0=0
	\end{aligned}
	\end{equation*}
 to complete our proof.
\end{proof}

We introduce some important convex functions that we will need to prove Theorem \ref{unbiasedness}.
Let
	\begin{equation*}
	\begin{aligned}
	D_{N_k}(\widetilde\gamma)&= \dfrac{1}{2}\widetilde\gamma' \widehat{\Sigma}_k \widetilde\gamma +\displaystyle\inf_{\widetilde{z}} \Big\{\dfrac{1}{2}(1-r_k)^{-1}r_k (\widetilde{z}-\widetilde\gamma+  \frac{1}{a_{N_k}}\hat{\Sigma}_{k}^{-1}\zeta_k)'\hat{\Sigma}_{k}  (\widetilde{z}-\widetilde\gamma + \frac{1}{a_{N_k}} \hat{\Sigma}_{k}^{-1}\zeta_k) \\
	&\;\;\;\;\;\;\;\;\;\;\;\;\;\;\;\;\;\;\;\;\;\;\;\;\;\;\;\;\;\;\;\;\;\;\;+  \dfrac{1}{(a_{N_k})^2}B_{s_{E_k}}(a_{N_k}\widetilde{z})\Big\},
	\end{aligned}
	\end{equation*}
 and let 
	\begin{equation}
	\label{P:D}
	\widetilde{\mathcal{P}}_{N_k}(\eta_0) = \displaystyle\sup_{\widetilde\gamma} \widetilde\gamma' \eta_0 -D_{N_k}(\widetilde\gamma),
	\end{equation}
	which is the convex conjugate for $D_{N_k}$.
 Similarly, let 
 \begin{equation*}
	\begin{aligned}
	\breve{D}_{N_k}(\widetilde\gamma)&= \dfrac{1}{2}\widetilde\gamma' \Sigma_k \widetilde\gamma +\displaystyle\inf_{\widetilde{z}} \Big\{\dfrac{1}{2}(1-r_k)^{-1}r_k (\widetilde{z}-\widetilde\gamma+  \frac{1}{a_{N_k}}\Sigma_{k}^{-1}\zeta_k)'\Sigma_{k}  (\widetilde{z}-\widetilde\gamma + \frac{1}{a_{N_k}} \Sigma_{k}^{-1}\zeta_k) \\
	&\;\;\;\;\;\;\;\;\;\;\;\;\;\;\;\;\;\;\;\;\;\;\;\;\;\;\;\;\;\;\;\;\;\;\;+  \dfrac{1}{(a_{N_k})^2}B_{s_{E_k}}(a_{N_k}\widetilde{z})\Big\},
	\end{aligned}
	\end{equation*}
 and let $\breve{\mathcal{P}}_{N_k}$ be the corresponding convex conjugate.
 
\begin{proof}\emph{Theorem \ref{unbiasedness}.}
	First, we note that our carved estimator is equal to
	\begin{equation}
	\label{est:eqn}
	\frac{1}{a_{N_k}}\sqrt{N_k}\widehat{\alpha}_k^{\text{\; carve}} = e_1'\widehat\Sigma^{-1}_{k}  \nabla D_{N_k}\left(\frac{1}{a_{N_k}}\sqrt{N_k} \hat\gamma_k\right),
	\end{equation}
	where $\nabla D_{N_k}$ is the gradient of $D_{N_k}$.
	This is because
	\begin{equation*}
	\begin{aligned}
	\nabla D_{N_k}(\widetilde\gamma) =  \widehat{\Sigma}_k \widetilde\gamma - (1-r_k)^{-1}r_k \hat{\Sigma}_{k} \left(z^{*}(\widetilde\gamma)-\widetilde\gamma+  \frac{1}{a_{N_k}}\hat{\Sigma}_{k}^{-1}\zeta_k\right),
	\end{aligned}
	\end{equation*}
	where 
	\begin{equation*}
	\begin{aligned}
	z^*(\widetilde\gamma) &= \text{arginf}_{z}\; \Big\{\dfrac{1}{2}(1-r_k)^{-1}r_k \left(\widetilde{z}-\widetilde\gamma+  \frac{1}{a_{N_k}}\hat{\Sigma}_{k}^{-1}\zeta_k\right)'\hat{\Sigma}_{k}  \left(\widetilde{z}-\widetilde\gamma + \frac{1}{a_{N_k}} \hat{\Sigma}_{k}^{-1}\zeta_k\right) \\
	&\;\;\;\;\;\;\;\;\;\;\;\;\;\;\;\;\;\;\;\;\;\;\;\;\;\;\; +  \dfrac{1}{(a_{N_k})^2}B_{s_{E_k}}(a_{N_k}\widetilde{z})\Big\},
	\end{aligned}
	\end{equation*}
	and
	$z^*\left(\frac{1}{a_{N_k}}\sqrt{N_k} \hat\gamma_k\right) =  \frac{1}{a_{N_k}}\sqrt{N_k}\widehat{z}_k$, where $\widehat{z}_k$ is defined in \eqref{optimizer}.

	Also, define
	$$\widetilde{\mathcal{P}}_{0,N_k}(\eta_0) =\frac{1}{(a_{N_k})^{2}}\log \mathbb{P}\left[\widehat{E}_k =  E_k, \;\text{sign}(\hatB_k) = s_{E_k}\right] + \frac{1}{2}\eta'_0 \Sigma^{-1}_k\eta_0, $$ 
	where $\eta_0= \Sigma_k\gamma_0$.
 Let $\lambda_{\text{max}}$ and $\lambda_{\text{min}}$ be
	the largest and smallest eigen-value of  $\hat{\Sigma}_{k}$. 
	Now, we have
	\begin{equation*}
	\begin{aligned}
	& \mathbb{E}\left[ \Big\|\frac{1}{a_{N_k}}\sqrt{N_k}(\widehat{\alpha}_k^{\text{\; carve}} - \alpha)\Big\|^2\Big\lvert  \widehat{E}_k = E_k,  \text{sign}(\hatB_k) =s_{E_k}\right] \\
	&
 \leq\dfrac{1}{\lambda^2_{\text{min}}(\widehat\Sigma_{k})}\mathbb{E}\left[ \Big\| \nabla\widetilde{\mathcal{P}}_{N_k}^{-1}\left(\frac{1}{a_{N_k}}\sqrt{N_k} \hat\gamma_k\right)- \widehat\Sigma_{k} \gamma_0\Big\|^2\Big\lvert  \widehat{E}_k = E_k,  \text{sign}(\hatB_k) =s_{E_k}\right]
	\end{aligned}
	\end{equation*}
	\begin{equation*}
	\begin{aligned}
	&\leq   \dfrac{C^2}{\lambda^2_{\text{min}}(\widehat\Sigma_{k})}\mathbb{E}\left[ \Big\|\frac{1}{a_{N_k}}\sqrt{N_k} \hat\gamma_k-  \nabla\widetilde{\mathcal{P}}_{N_k}(\widehat\Sigma_{k} \gamma_0)\Big\|^2\Big\lvert  \widehat{E}_k = E_k,  \text{sign}(\hatB_k) =s_{E_k}\right] \\
	&\leq  \dfrac{C^2}{\lambda^2_{\text{min}}(\widehat\Sigma_{k})}\Bigg\{\mathbb{E}\left[ \Big\|\frac{1}{a_{N_k}}\sqrt{N_k}\hat\gamma_k-  \nabla\widetilde{\mathcal{P}}_{0,N_k}(\Sigma_{k} \gamma_0)\Big\|^2\Big\lvert  \widehat{E}_k = E_k,  \text{sign}(\hatB_k) =s_{E_k}\right]  \\
	&\;\;\;\;\;\;\;\;\;\;\;\;\;\;\;\;\;\;\;\;\;\;\;\;\;\;\;\;\;\;\;\; + \mathbb{E}\left[\Big\| \nabla\breve{\mathcal{P}}_{N_k}(\Sigma_{k} \gamma_0)-  \nabla\widetilde{\mathcal{P}}_{N_k}(\widehat\Sigma_{k} \gamma_0)\Big\|^2\Big\lvert  \widehat{E}_k = E_k,  \text{sign}(\hatB_k) =s_{E_k}\right] \\
	&\;\;\;\;\;\;\;\;\;\;\;\;\;\;\;\;\;\;\;\;\;\;\;\;\;\;\;\;\;\;\;\;\;+ \Big\| \nabla\breve{\mathcal{P}}_{N_k}(\Sigma_{k} \gamma_0)-  \nabla\widetilde{\mathcal{P}}_{0,N_k}(\Sigma_{k} \gamma_0)\Big\|^2\Bigg\}.
	\end{aligned}
	\end{equation*}
	The second display uses \eqref{P:D} and \eqref{est:eqn} to deduce:
	\begin{equation*}
	\begin{aligned}
	e_1'\widehat\Sigma^{-1}_{k} \nabla\widetilde{\mathcal{P}}_{N_k}^{-1}\left(\frac{1}{a_{N_k}}\sqrt{N_k} \hat\gamma_k\right) &= e_1'\widehat\Sigma^{-1}_{k} \nabla D_{N_k}\left(\frac{1}{a_{N_k}}\sqrt{N_k} \hat\gamma_k\right)= \frac{1}{a_{N_k}}\sqrt{N_k}\widehat{\alpha}_k^{\text{\; carve}}.
	\end{aligned}
	\end{equation*} 
	The third display uses the fact that $\nabla\widetilde{\mathcal{P}}_{N_k}^{-1}$ is Lipschitz with constant $C$ where $C= (1-r_k)^{-1} \lambda_{\text{max}}(\hat{\Sigma}_{k})$.
 The last display follows after using the standard triangle inequality for decomposing the preceding expectation. 
 For the decomposition shown in the last display, it is easy note that the expectation in the first term is $O(1)$.
    Observe that $\nabla\widetilde{\mathcal{P}}_{N_k}(\cdot)$, in the second display, is  Lipschitz, because it is the gradient of the conjugate of a strongly convex function.
    Using the fact that $\widehat{\Sigma}_k-\Sigma_k= o_p(1)$ in our i.i.d. framework, jointly with the condition in Assumption \ref{operator:norm:bdd:0}, we deduce that the expected value of the second term is $O(1)$. 
    The third term converges to $0$ by using the convexity of $\widetilde{\mathcal{P}}_{0,N_k}$ together with the pointwise convergence shown by Proposition \ref{conv:log:partition}. 
	This completes our proof.
\end{proof}

We turn to deriving a feasible value for the observed Fisher information matrix in order to estimate the variance of the carved estimator. To this end, the limiting value for the probability of selection in Proposition \ref{conv:log:partition} gives rise to an approximate log-partition function in terms of $\eta_0= \Sigma_k\gamma_0$, the natural parameters in the asymptotic distribution of the refitted least squares estimation after conditioning on the selection event. 

Letting $$\sqrt{N}_k \eta = a_{N_k}\eta_0,$$
observe that the approximate log-partition function based on Proposition \ref{conv:log:partition} is given by:
\begin{equation}
\label{approx:part}
\begin{aligned}
& \frac{a^2_{N_k}}{2}\eta_0'  \Sigma^{-1}_{k} \eta_0 -a_{N_k}^2\cdot\displaystyle\inf_{\widetilde{\gamma}, \widetilde{z}}\Big\{ \dfrac{1}{2}(\widetilde{\gamma}-  \Sigma^{-1}_{k} \eta_0)' \Sigma_{k} (\widetilde{\gamma}-  \Sigma^{-1}_{k} \eta_0) \\
&\;\;\;\;\; +  \dfrac{1}{2}(1-r_k)^{-1}r_k (\widetilde{z}-\widetilde{\gamma}+  \frac{1}{a_{N_k}}\Sigma_{k}^{-1}\zeta_k)'\Sigma_{k}  (\widetilde{z}-\widetilde{\gamma} + \frac{1}{a_{N_k}} \Sigma_{k}^{-1}\zeta_k) +  \dfrac{1}{(a_{N_k})^2}B_{s_{E_k}}(a_{N_k}z)\Big\}\\
&= \frac{N_k}{2}\eta' \Sigma^{-1}_{k} \eta -\displaystyle\inf_{\gamma, z}\Big\{ \dfrac{1}{2}(\sqrt{N_k}\gamma-  \sqrt{N_k}\Sigma^{-1}_{k} \eta)' \Sigma_{k} (\sqrt{N_k}\gamma-  \sqrt{N_k}\Sigma^{-1}_{k} \eta) \\
&+  \dfrac{1}{2}(1-r_k)^{-1}r_k (\sqrt{N_k}z-\sqrt{N_k}\gamma+  \Sigma_{k}^{-1}\zeta_k)'\Sigma_{k}  (\sqrt{N_k}z-\sqrt{N_k}\gamma +  \Sigma_{k}^{-1}\zeta_k) \\
&\;\;\;\;\;\;\;\;\;\;\;\;\;\;\;\;\;\;\;\;\;\;\;\;\;\;\;\;\;\;\;\;\;\;\;\;\;\;\;+  B_{s_{E_k}}(\sqrt{N_k}z)\Big\}.
\end{aligned}
\end{equation}
The last display is derived  by 
reparameterizing $a_{N_k} \widetilde\gamma = \sqrt{N_k}\gamma$, $a_{N_k} \widetilde{z} = \sqrt{N_k}z$.
Denote by $\mathcal{P}_{N_k}(\eta)$ the approximate log-partition function in the last display, which we use to obtain an operational expression for the observed Fisher information matrix next.
We plug in the sample estimate for the unknown covariance $\Sigma_k$ to estimate the value of the observed Fisher information matrix.
Proposition \ref{Fisher:info} outlines the derivation of this estimator. 
The symbol $e_1$ in the claim denotes a $q_k + s$ dimensional block diagonal matrix, where the first $s$-dimensional matrix along the diagonal is the identity matrix with $s$ columns and the second $q_k$-dimensional matrix is a matrix of all zeros.

\begin{proposition}
	\label{Fisher:info}
	Based on the approximate value for the log-partition function in \eqref{approx:part}, the inverse of the observed Fisher information matrix assumes the following expression:
	\begin{equation*}
	\begin{aligned}
	e_1'\left((1-r_k)^{-1}\hat{\Sigma}^{-1}_{k}  - (1-r_k)^{-2}r^2_k\left((1-r_k)^{-1}r_k\hat{\Sigma}_{k} + \nabla^2 B_{s_{E_k}}(\sqrt{N_k} \widehat{z}_k)\right)^{-1}\right)e_1.
	\end{aligned}
	\end{equation*}
\end{proposition}

\begin{proof}
	Using the approximate expression for the log-partition function in \eqref{approx:part}, the observed Fisher information submatrix is equal to
	$$e_1' \hat{\Sigma}_{k} \nabla^2 \mathcal{P}_{N_k}( \widehat\Sigma_k\widehat{\gamma}_k^{\text{\; carve}}) \hat{\Sigma}_{k} e_1 = e_1' \hat{\Sigma}_{k}  \nabla \gamma^*\hat{\Sigma}_{k} e_1$$
	where $ \gamma^*$ equals
	\begin{equation}
	\label{opt:fi}
	\begin{aligned}
	& \displaystyle\arg\sup_{\gamma} \sqrt{N_k}\gamma' \hat\Sigma_k\sqrt{N_k}\widehat{\gamma}_k^{\text{\; carve}} -\dfrac{1}{2}\sqrt{N_k}\gamma' \widehat{\Sigma}_k \sqrt{N_k}\gamma \\
	&\;\;\;\;\;\;-\displaystyle\inf_{z} \Big\{\dfrac{1}{2}(1-r_k)^{-1}r_k (\sqrt{N_k}z-\sqrt{N_k}\gamma+ \hat{\Sigma}_{k}^{-1}\zeta_k)'\hat{\Sigma}_{k}  (\sqrt{N_k}z-\sqrt{N_k}\gamma + \hat{\Sigma}_{k}^{-1}\zeta_k) \\
	&\;\;\;\;\;\;\;\;\;\;\;\;\;\;\;\;\;\;\;\;\;\;\;\;\;\;\;\;\;\;\;\;\;\;\;\;\;\;\;+ B_{s_{E_k}}(\sqrt{N_k}z)\Big\}.
	\end{aligned}
	\end{equation}
	Solving \eqref{opt:fi} gives us
	$$ \hat\Sigma_k\sqrt{N_k}\widehat{\gamma}_k^{\text{\; carve}}= \widehat{\Sigma}_k   \sqrt{N_k}\gamma^*+ (1-r_k)^{-1}r_k \hat{\Sigma}_{k} \left(\sqrt{N_k}\gamma^*- \hat{\Sigma}_{k}^{-1}\zeta_k- \sqrt{N_k}\widehat{z}_k\right),$$ 
	where
	$$
	(1-r_k)^{-1}r_k\hat{\Sigma}_{k}  \left(\sqrt{N_k} \widehat{z}_k-\sqrt{N_k}\gamma^*+   \hat{\Sigma}_{k}^{-1}\zeta_k\right)  + \nabla B_{s_{E_k}}(\sqrt{N_k} \widehat{z}_k)= 0.
	$$
	The former equations yields us $\gamma^*= \widehat{\gamma}_k$, using the expression for our carved estimator.
	Taking further derivatives, we obtain
	$$\nabla_{\gamma^*} \widehat{z}_k = \left((1-r_k)^{-1}r_k\hat{\Sigma}_{k} + \nabla^2 B_{s_{E_k}}(\sqrt{N_k}  \widehat{z}_k)\right)^{-1}(1-r_k)^{-1}r_k\hat{\Sigma}_{k},$$
	\begin{equation*}
	\begin{aligned}
	&\nabla \gamma^*  = \left((1-r_k)^{-1} \hat{\Sigma}_{k}  - (1-r_k)^{-1}r_k\hat{\Sigma}_{k}\nabla_{\gamma^*} \widehat{z}_k\right)^{-1}\\
	&= \Big((1-r_k)^{-1} \hat{\Sigma}_{k}  - (1-r_k)^{-1}r_k\hat{\Sigma}_{k} \left((1-r_k)^{-1}r_k\hat{\Sigma}_{k} + \nabla^2 B_{s_{E_k}}(\sqrt{N_k} \widehat{z}_k)\right)^{-1}(1-r_k)^{-1}r_k\hat{\Sigma}_{k}\Big)^{-1}.
	\end{aligned}
	\end{equation*}
	Then, the inverse for the Fisher information matrix is given by:
	\begin{equation*}
	\begin{aligned}
	(1-r_k)^{-1}\hat{\Sigma}^{-1}_{k}  - (1-r_k)^{-2}r^2_k\left((1-r_k)^{-1}r_k\hat{\Sigma}_{k} + \nabla^2 B_{s_{E_k}}(\sqrt{N_k} \widehat{z}_k)\right)^{-1},
	\end{aligned}
	\end{equation*}
 which directly leads to our claim in the Proposition.
\end{proof}

\section{Carved Estimator for Logistic Regression}
\label{Appendix:A}

Fixing some notations, let 
$$p(x) = \dfrac{\exp(x)}{1+ \exp(x)}$$
for $x\in \mathbb{R}$.

Suppose that we observe a binary response variable which assumes values in $\{0,1\}$. 
Introducing some notations, we let
$$p(x) = \dfrac{\exp(x)}{1+ \exp(x)}$$
for $x\in \mathbb{R}$.
Let $\pi(D_k\alpha + X_k \beta) \in \mathbb{R}^{n_k}$ be the vector with 
$$p\left((D_k^{(i)})'\alpha+ (X_k^{(i)})' \beta\right)$$ 
as its $i^{\text{th}}$ component.
We let
\begin{equation*}
\begin{aligned}
\ell(Y_k, \pi(D_k\alpha + X_k \beta)) &= -\sum_{i=1}^{n_k} \Big\{ Y_k^{(i)}\log p\left((D_k^{(i)})'\alpha+ (X_k^{(i)})' \beta\right) \\
&\;\;\;\;\;\;+ (1-Y_k^{(i)}) \log\left(1- p\left((D_k^{(i)})'\alpha+ (X_k^{(i)})' \beta\right)\right)\Big\}
\end{aligned}
\end{equation*}
denote the logistic loss based on $n_k$ samples of the $k^{\text{th}}$ existing study.

In each existing study, we apply the logistic regression with an added LASSO penalty to estimate
\begin{equation}\label{eq:logistic:lasso}
\widehat{\gamma}^{\text{\;(L)}}_{k} = \begin{pmatrix} \widehat{\alpha}^{\text{\;(L)}}_{k} \\ \widehat{\beta}^{\text{\;(L)}}_{k}\end{pmatrix}=\underset{\alpha\in\mathbb{R}^s,\beta\in \mathbb{R}^{p_k} }{\arg\min}\left\{ \frac{1}{r_k\sqrt{N_k}} \ell(Y_k, \pi(D_k\alpha + X_k \beta))
+ || \Lambda_k \beta ||_1 \right\}, 
\end{equation}
where 
$\widehat{\alpha}^{\text{\;(L)}}_{k}\in \mathbb{R}^s$, $\widehat{\beta}^{\text{\;(L)}}_{k} \in \mathbb{R}^{q_k}$ for $k=1,2,\ldots, K$.
Denote by $E_k$ the support set of $\widehat{\beta}_k^{\text{\;(L)}}$. 

Next, we identify the summary statistics that we require from each existing study, and give the form of our estimator.
Consider the following summary information from existing study $k$. 
\begin{enumerate}
\setlength\itemsep{1em}
	\item \textit{Summary from model selection}: the support set of the LASSO estimator $E_k$, the penalty weights $\Lambda_{ E_k}$, and the signs of the (nonzero) LASSO estimator for the selected covariates $s_{E_k} = \text{sign}(\widehat{\beta}_k^{\text{\;(L)}})$. \label{summary:sel}
        \item \textit{Summary from the first two moments in the selected model}: the MLE in the selected model $E_k$
        $$
	\widehat{\xi}_k := \left[ \begin{array}{cc}
	\widehat{\xi}_{k,01} \\
	\widehat{\xi}_{k,02} 
	\end{array} \right]$$
        and the estimated information matrix in the same model
	\begin{align*}
	\hat{\Xi}_{k} :=  \ \left[
	\begin{array}{cc}
	\hat{\Xi}_{k, 11} & \hat{\Xi}_{k, 12}\\
	\hat{\Xi}_{k, 21} & \hat{\Xi}_{k, 22} 
	\end{array} \right]   =  & \frac{1}{n_k} \left[
	\begin{array}{cc}
	D_k'W_{k}(\widehat{\xi}_k)D_k & D_k'W_{k}(\widehat{\xi}_k)X_{k, E_k}\\
	X_{k, E_k}'W_{k}(\widehat{\xi}_k)D_k & X_{k,  E_k}'W_{k}(\widehat{\xi}_k)X_{k, E_k}
	\end{array} \right],
	\end{align*} 
    where 
    $W_k(\widehat{\xi}_k) \in \mathbb{R}^{n_k \times n_k}$ denotes a diagonal matrix with
$$\pi_i\left(\begin{bmatrix} D_k & X_{k,E_k} \end{bmatrix}\widehat{\xi}_k\right) \left(1- \pi_i\left(\begin{bmatrix} D_k & X_{k,E_k} \end{bmatrix}\widehat{\xi}_k\right)\right)$$
in its $i^{\text{th}}$ diagonal entry.
	Note that $\hat{\xi}_k$ and $\hat{\Xi}_{k} $ are partitioned along the same dimensional structure as before.
	\label{summary:data}
\end{enumerate}

Let $\hat\xi$ be the MLE in the selected model, using observations in the validation study, and let $\hat{\Sigma}= D'W(\hat\xi)D$ be the (corresponding) estimated information matrix in the same model.
Denote by $W(\widehat{\xi})\in \mathbb{R}^{n \times n}$ a diagonal matrix with
$$\pi_i\left(\begin{bmatrix} D & X_{E_k} \end{bmatrix}\widehat{\xi}\right) \left(1- \pi_i\left(\begin{bmatrix} D & X_{E_k} \end{bmatrix}\widehat{\xi}\right)\right)$$
in its $i^{\text{th}}$ diagonal entry.
First, we aggregate the estimated information matrix as
\begin{equation}
\hat{\Sigma}_{k} =  n_k\hat{\Xi}_{k}  + \begin{pmatrix} D & X_{E_k} \end{pmatrix}' W(\widehat{\xi})\begin{pmatrix} D & X_{E_k} \end{pmatrix},
\label{sample:pooled:covariance}
\end{equation}
and then we compute the aggregated MLE as
$$
\hat\gamma_k=\begin{pmatrix} (\hat\alpha_k)' & (\hat\beta_k)'\end{pmatrix}' =\hat{\Sigma}_{k}^{-1} \left( n_k\hat{\Xi}_{k} \widehat{\xi}_k + \hat{\Sigma}\widehat{\xi}\right).
$$
Hereon, our carved estimator takes the exact same form as \eqref{eq:gamma-carve}. 
That is, we solve the optimization in \eqref{optimizer} by using the aggregated MLE $\hat\gamma_k$, and information matrix $\hat{\Sigma}_{k}$.

Our theory follows the exact steps as before once we define
our randomization variable as
\begin{equation}
\begin{aligned}
\omega_k &= \begin{pmatrix} (\omega_{k,1})' & (\omega_{k,2})' \end{pmatrix}'; \; \omega_{k,1}\in \mathbb{R}^s, \omega_{k,2}\in \mathbb{R}^{p_k}\\
&=\dfrac{\partial}{\partial\gamma}\Bigg(\frac{1}{\sqrt{N_k}} \ell(Y_k, \pi(D_k\alpha + X_k \beta)) + \frac{1}{\sqrt{N_k}} \ell(Y, \pi(D\alpha + X \beta))\\
&\;\;\;\;\;\;\;\;\;\;\;\;\;\;\;\;\;\;-\frac{1}{r_k\sqrt{N_k}} \ell(Y_k, \pi(D_k\alpha + X_k \beta))\Bigg)\Bigg\lvert_{\hatG_k},
\label{randomization:gen}
\end{aligned}
\end{equation}
and set
$$\hat{\Gamma}_{k} = \dfrac{1}{N_k}\left\{X'_{k,E_k^c}(Y_k-\pi(D_k\hat\alpha_k + X_k \hat\beta_k)) + X'_{E_k^c}(Y-\pi(D\hat\alpha_k + X \hat\beta_k)) \right\}.$$
Then, we have the asymptotic linear representation as stated in Lemma \ref{limiting:Gaussianity}, and the rest of the theory proceeds as before.

\end{document}